\documentclass{lmcs}
\usepackage[utf8]{inputenc}
\pdfoutput=1

\usepackage{lastpage}
\lmcsdoi{18}{3}{32}
\lmcsheading{}{\pageref{LastPage}}{}{}%
{Apr.~27,~2021}{Sep.~12,~2022}{}


\usepackage{float}
\floatstyle{boxed}
\restylefloat{table}
\usepackage[only=llparenthesis,rrparenthesis,llbracket,rrbracket]{stmaryrd}
\usepackage{tikz}
\usetikzlibrary{cd,positioning,decorations.text}
\tikzcdset{scale cd/.style={every label/.append style={scale=#1},cells={nodes={scale=#1}}}}
\usepackage{needspace}
\usepackage{proof}
\usepackage{amsmath,amsfonts,amssymb,amsthm,mathtools}
\allowdisplaybreaks%
\usepackage{xspace}
\usepackage{hyperref}
\hypersetup{
  colorlinks=true,
  linkcolor=blue,
  anchorcolor=blue,
  citecolor=blue,
  urlcolor=black,
  filecolor=blue,
  breaklinks=true
}

\newcommand\OC{Lambda-${\mathcal{S}}_1$\xspace}
\newcommand\ket[1]{\ensuremath{|#1\rangle}}
\newcommand\teoremaunitarias{\cite[Theorem IV.12]{DiazcaroGuillermoMiquelValironLICS19}}
\newcommand\Definible{\mathsf{Def}}
\newcommand\parts[1]{\ensuremath{{\mathcal P}_{\!\!*}({#1})}}

\newcommand\xrecap[4]{\noindent{\bf #1~\ref{#3}} (#2){\textbf{.}}{\textit{#4}}}
\newcommand\home[2]{[#1,#2]}
\newcommand\I{I}
\newcommand\Id{\mathsf{Id}}
\newcommand\sem[1]{\left\llbracket{#1}\right\rrbracket}
\newcommand\semR[1]{\llparenthesis{#1}\rrparenthesis_{\!{}_R}}
\def\<{\langle}
\def\>{\rangle}
\newcommand\Void{*} % void object

\newcommand\Pair[2]{(#1,#2)} % pairing construct
\newcommand\Lam[2]{\lambda#1\,{.}\,#2} % lambda abstraction
\newcommand\letkeyword{\mathtt{let}}
\newcommand\inkeyword{\mathtt{in}}

\newcommand\LetP[4]{\letkeyword~\Pair{#1}{#2}=#3~\inkeyword~#4}
\newcommand\inleftkeyword{\mathtt{inl}}
\newcommand\Inl[1]{\inleftkeyword(#1)}

\newcommand\inrightkeyword{\mathtt{inr}}
\newcommand\Inr[1]{\inrightkeyword(#1)}

\newcommand\matchkeyword{\mathtt{match}}
\newcommand\Match[5]{\matchkeyword~#1~%
  \{\Inl{#2}\mapsto#3~|~\Inr{#4}\mapsto#5\}}

\newcommand\FV{\textit{FV}}
\newcommand\C{\mathbb{C}}
\newcommand\lra{\longrightarrow}
\newcommand\U{\ensuremath{\mathbb U}}
\newcommand\s[1]{\ensuremath{\mathsf{#1}}}
\newcommand\So{\ensuremath{{\mathcal S}_1}}

\newcommand\Set{\s{Set}}

\newcommand\Arr[1]{\s{Arr}(#1)}

\newcommand\HomV[1]{\s{Hom}_{\VecV}(#1)}
\newcommand\HomS[1]{\s{Hom}_{\SetV}(#1)}

\newcommand\Ob[1]{\s{Ob}(#1)}
\newcommand\Span[1]{\ensuremath{{S}{#1}}}
\newcommand\Forg[1]{\ensuremath{{U}{#1}}}

\newcommand\HomSfiV{\HomS{f,\Forg V}}
\newcommand\HomHsfV{\HomV{\Span f,V}}
\newcommand\HomHsAg{\HomV{\Span A,g}}
\newcommand\HomSAig{\HomS{A,\Forg g}}

\newcommand\Val{{\s V}}
\newcommand\ValD{\vec{\s V}}
\newcommand\VecV{\s{SVec}_{\ValD}}
\newcommand\SetV{\Set_{\ValD}}

\newcommand\Hil{\mathcal H_{\ValD}}
\newcommand\xlra[1]{\xrightarrow{#1}}

\begin{document}

\title[Lambda-\texorpdfstring{\So}{S1} and its Categorical Model]{Quantum Control in the Unitary Sphere:\texorpdfstring{\\}{} Lambda-\texorpdfstring{\So}{S1} and its Categorical Model}

\author[A.~D\'{\i}az-Caro]{Alejandro D\'{\i}az-Caro\rsuper{a,b}}
\address{Instituto de Ciencias de la Computaci\'on, CONICET--Universidad de Buenos Aires. % chktex 8
Buenos Aires, Argentina}
\address{Depto.~de Ciencia y Tecnolog\'{\i}a, Universidad Nacional de Quilmes.
Bernal, Buenos Aires, Argentina}
\email{adiazcaro@icc.fcen.uba.ar} %optional
\thanks{Partially funded by PIP 11220200100368CO, PICT-2019-1272, 21STIC10 Qapla', PUNQ 1342/19 and ECOS-Sud A17C03 QuCa.}	%optional % chktex 8

\author[O.~Malherbe]{Octavio Malherbe\rsuper{c,d}}	%optional
\address{Instituto de Matem\'atica y Estad\'{\i}stica ``Rafael Laguardia'', FIng, Universidad de la Rep\'ublica.
Montevideo, Uruguay}
\address{Depto.~de Matem\'atica y Aplicaciones, CURE, Universidad de la Rep\'ublica.
Maldonado, Uruguay}
\email{malherbe@fing.edu.uy}  %optional
%\thanks{thanks 2, optional.}	%optional

\begin{abstract}
  In a recent paper, a realizability technique has been used to give a semantics of a quantum lambda calculus. Such a technique gives rise to an infinite number of valid typing rules, without giving preference to any subset of those. In this paper, we introduce a valid subset of typing rules, defining an expressive enough quantum calculus.  Then, we propose a categorical semantics for it. Such a semantics consists of an adjunction between the category of distributive-action spaces of value distributions (that is, linear combinations of values in the lambda calculus), and the category of sets of value distributions.
\end{abstract}

\maketitle

\section{Introduction}\label{introduction}
In quantum programming languages, the control flow of programs divides models
in two classes. On the one hand, there is the model of the QRAM~\cite{Knill04},
or classical control~\cite{SelingerMSCS04}. The classical control refers to a
scheme where the quantum operations are performed in a specialized device,
known as QRAM, attached to a classical computer, which instructs the device
which operations to apply over which qubits.  It is the more realistic and
practical scenario.  In this model, the quantum operations are given by a
series of ``black boxes''. An example of this is the quantum lambda
calculus~\cite{SelingerValironMSCS06}, as well as several high-level quantum
programming languages such as Quipper~\cite{GreenEtAlPLDI13} and
QWIRE~\cite{PaykinRandZdancewicPOPL17}.  The kind of problems that this model
dealt with is, for example, to forbid cloning unknown qubits, since the
non-cloning theorem states that there is no universal cloning machine.

On the other hand, there is the model of quantum control, with a parallel
agenda. The ultimate motivation of this model is to extend the Curry-Howard
isomorphism relating type theory with logics, to the quantum case. Indeed,
there is a long line of research on Quantum Logic started by the pioneer work
of Birkhoff and Von Neumann in the 30's~\cite{BirkhoffVonNeumann}. However, the
connection from this logic to a lambda calculus is unknown.

One of the first works on the quantum control approach is the development of
QML~\cite{AltenkirchGrattageLICS05}, where the quantum control is expressed by
the quantum if ``$\s{if}^\circ$'' which, given a superposition of $\ket 0$ and
$\ket 1$, produces a superposition of its two output branches. However, a
superposition of the form $\alpha.\ket 0+\beta.\ket 1$ is a valid qubit only if
its norm is equal to $1$, so, if $|\alpha|^2+|\beta|^2=1$. Therefore,
superposing the two output branches would be valid, only if the norm of this
term is equal to $1$. Therefore, for example, the term $\s{if}^\circ\
\alpha.\ket 0+\beta.\ket 1\ \s{then}\ s\ \s{else}\ t$ is valid only if $s$ and
$t$ are orthogonal, and so it preserves the norm. Thus, QML introduced a notion
of norm for a small subset of terms. Indeed, consider the type Bool, with its
two orthogonal values $\s{true}$ and $\s{false}$. To detect if $t:Bool$ and
$r:Bool$ are orthogonal, means to reduce $t$ and $r$ and to compare them. So
the orthogonality question in typing have been an open question for many years,
until the work~\cite{DiazcaroGuillermoMiquelValironLICS19}, which provided a
partial answer.

The long path to this partial answer started at
Lineal~\cite{ArrighiDowekRTA08,ArrighiDowekLMCS17}, which is an untyped
extension to the lambda calculus allowing for linear combinations of terms.
This way, if $s$ and $t$ are two terms, so is its formal linear combination
$\alpha\cdot s+\beta\cdot t$, with $\alpha,\beta\in\mathbb C$.  Unitary
matrices are not expressed as given black boxes, but they can be constructed.

Quantum programs can be expressed in Lineal, except for the quantum
measurement, which is left out of the system. However, Lineal is not restricted
to only quantum programs.  In particular, the vectors are not ensured to be of
norm $1$, since it would require checking for orthogonality between vectors.
Neither functions are proved to behave as isometries, as needed by the quantum
theory. One main feature of Lineal, although, is the fact that all the
functions, even if they are not forced to be linear or isometries, are treated
linearly: if a function $\lambda x.s$ is applied to a formal linear combination
$\alpha\cdot v+\beta\cdot w$, it distributes linearly as follows:
\[
  (\lambda x.s)(\alpha\cdot v+\beta\cdot w)\lra\alpha\cdot(\lambda
  x.s)v+\beta\cdot(\lambda x.s)w
\]
generalising the quantum-if from QML\@.

A drawback in taking all functions as linear, is that adding measurement was
not trivial, since a measurement is not a linear operation: If $M$ is a
measurement operator, $M(\alpha\cdot v+\beta\cdot w)$ does not behave as
$\alpha\cdot Mv+\beta\cdot Mw$.

Lambda-$\mathcal S$~\cite{DiazcaroDowekTPNC17,DiazcaroDowekRinaldiBIO19} is a
typed lambda calculus based on Lineal, mainly focused on adding measurement to
the calculus. Instead of treating all functions as linear, its types enforce
linearity when the argument is a superposition, and allow for duplication when
it is not. This is done by labelling superpositions with a modality $\mathcal
S$.  Any term typed by $\mathcal S$ is treated linearly, so only basis terms
are duplicable. It is argued to be somehow the dual to Intuitionistic Linear
Logic, where duplicable terms are marked (by a $!$). Indeed, in~\cite{DiazcaroMalherbeLSFA18,DiazcaroMalherbe20,DiazcaroMalherbeACS20} a
categorical model for Lambda-$\mathcal S$ has been proposed, obtained by a
monoidal monad determined by a monoidal adjunction $(S,m)\dashv(U,n)$ and
interpreting $\mathcal S$ as the monad $US$---exactly the opposite to the $!$
of linear logic, which in the literature is often interpreted as the comonad
$SU$ (see~\cite{MelliesHAL03}). This implies that on the one hand there is a
tight control of the Cartesian structure of the model, and on the other hand
the world of superpositions lives inside the classical word, i.e.~determined
externally by classical rules until one decides to explore it. This is given by
the following composition of maps:
\[
  USA\times USA\xlra{n} U(SA\otimes SA)\xlra{Um} US(A\times A)
\]
that allows us to operate in a monoidal structure explicitly allowing the
algebraic manipulation and then to return to the Cartesian product. This is
different from linear logic, where the $!$ stops any algebraic manipulation,
i.e.~$({!}A)\otimes ({!}A)$ is a product inside a monoidal category. A concrete
example is an adjunction between the categories $\Set$ of sets and $\s{Vec}$ of
vector spaces~\cite{DiazcaroMalherbeLSFA18,DiazcaroMalherbe20}.

The problem of orthogonality has been finally addressed
in~\cite{DiazcaroGuillermoMiquelValironLICS19}, which provides another type
system for Lineal, ensuring superpositions to be in the unitary sphere
$\mathcal S_1$ (that is, \mbox{norm-$1$} vectors).  It also characterizes
isometries via a specific type.  On this system, measurement has been left out
of the equation, since Lambda-$\mathcal S$ already showed how to add
measurement on Lineal, so, it is already known how to add measurement and it is no longer a problem.
This system ensuring norm-$1$ vectors and characterizing isometries has been
obtained by means of realizability
techniques~\cite{KleeneJSL45,Vonoosten08,KrivinePS09,MiquelTLCA11}. Instead of
deriving a computational meaning of proofs once the type system is set up, the
idea of realizability is to consider the type system as a by-product of the
operational semantics---programs are then potential realizers of types. For
example, a program behaving as the identity will be a realizer of $A\rightarrow
A$, regardless of its inner structure. Realizability is a powerful and modular
framework amenable to many systems (see~\cite{BrunelPhD14}). So, one
particularity is that the typing rules are probable lemmas (any typing rule
conforming the semantics, is a valid rule), hence, the set of rules is
potentially infinite. On this scheme, there is a modality $\sharp$, with a
similar behaviour to the $\mathcal S$ of Lambda-$\mathcal S$. The claimed main
goal of this system has been to solve the long-standing issue of how to ensure
\mbox{norm-$1$} superpositions, and characterize unitary functions.

The goal of the present paper is to extract a (finite) fixed type system
following the realizability
semantics~\cite{DiazcaroGuillermoMiquelValironLICS19}, a calculus we call \OC,
ensuring \mbox{norm-$1$} superpositions. We also give a categorical model for
this calculus.

Hence, the main contributions are twofold. On the one hand, we give the
definition of \OC, which is not trivial
since~\cite{DiazcaroGuillermoMiquelValironLICS19} provides only a method to
produce an infinite type system. On the other hand, the second main
contribution is the categorical model, which has some common grounds with the
concrete model of Lambda-$\mathcal
S$~\cite{DiazcaroMalherbeLSFA18,DiazcaroMalherbe20}, however, the chosen
categories this time are not $\Set$ and $\s{Vec}$, but categories that use the
fact that values in our calculus form a distributive-action space (an algebraic
  structure similar to a vector space, where its additive structure is a
semi-group). Summarising, the main novelty and contribution of this paper is
presenting a model for quantum computing in the quantum control paradigm, which
we show to be complete on qubits (Theorem~\ref{thm:completenessonqubits}).

We left the measurement operator out of the obtained system, only for the sake
of simplicity.
    The inclusion of the measurement operator was a problem for Lineal.
    However, after Lambda-$\mathcal
    S$~\cite{DiazcaroDowekTPNC17,DiazcaroDowekRinaldiBIO19} tackled this
    problem by considering linear types and the modality $\mathcal S$, it is no
    longer a problem. Since Lambda-$\So$\ uses the same modality (here written
    $\sharp$), the inclusion or not of a measurement operator does not suppose
    a challenge any more.
So, adding a measurement operator as the one from Lambda-$\mathcal
S$ is not difficult, since the problems were already solved on that system, and
\OC follows the same line. However, adding a measurement operator implies to
have a probabilistic rewrite system, which demands an extra monad (the
probabilistic monad) to be added to the model.  While this addition would be
easy (cf.~\cite{DiazcaroMalherbeLSFA18,DiazcaroMalherbe20}), it introduces superfluous complexity to the system making its model less clear.

In Section~\ref{sec:LambdaS1} we introduce the calculus \OC. We prove its main % chktex 13
correctness properties such as progress, subject reduction, and strong
normalization in~Section~\ref{sec:ExpCorr}. In Section~\ref{sec:expressivity}
we show the expressiveness of \OC, which includes the simply typed lambda
calculus with addition and pairs, plus the isometries. In
Section~\ref{sec:DenSem} we introduce its categorical model. In
Section~\ref{sec:SC} we prove the soundness of the model, and the completeness
of the type $\sharp(\U+\U)$, which corresponds to $\C^2$ (the type of qubits).
We conclude in Section~\ref{sec:Conclusion} with some final remarks.

\section{\texorpdfstring{\OC}{Lambda-S1}}\label{sec:LambdaS1}
\subsection{Terms}
In Table~\ref{tab:Syntax} we give the grammar of terms for \OC. Those fall into % chktex 13
two categories: ``pure'' and ``distributions'', and those in two subcategories,
of terms and values. The idea is that a distribution is a linear combination of
a pure terms. Formally, the set of distributions is equipped with a congruence
$\equiv$ that is generated from the 6~rules of Table~\ref{tab:cong-term}. We
say that a term $\vec t:=\sum_{i=1}^n\alpha_i\cdot t_i$ is given in canonical
form when $t_j\neq t_k$ for all $j,k$.

\begin{table*}[t]
  \[
    \begin{array}{l@{\ }r@{\ }c@{\ }l}
      \parbox[c]{1cm}{\centering\tiny Pure values}&v,w&:=& x \mid \Lam{x}{\vec{s}} \mid \Void \mid \Pair{v_1}{v_2}\mid \Inl{v} \mid \Inr{v}\\[6pt]
      \parbox[c]{1cm}{\centering\tiny Pure terms}&s,t&:=& v \mid s\,t \mid t;\vec{s} \mid \LetP{x_1}{x_2}{t}{\vec{s}} \mid\Match{t}{x_1}{\vec{s}_1}{x_2}{\vec{s}_2}\\[6pt]
      \parbox[c]{1cm}{\centering\tiny Value distrib.}&\vec{v},\vec{w}&:=& v \mid \vec{v}+\vec{w} \mid \alpha\cdot\vec{v}\qquad\hfill(\alpha\in\C)\\[6pt]
      \parbox[c]{1cm}{\centering\tiny Term distrib.}&\vec{s},\vec{t}&:=& t \mid \vec{s}+\vec{t} \mid \alpha\cdot\vec{t}\qquad\hfill(\alpha\in\C)
    \end{array}
  \]
  \caption{Grammar of terms}%
  \label{tab:Syntax}
\end{table*}

\begin{table*}[t]
  \[
    \begin{array}{r@{\ }c@{\ }l@{\qquad}r@{\ }c@{\ }l@{\qquad}r@{\ }c@{\ }l}
      \vec{t}_1+\vec{t}_2&\equiv&\vec{t}_2+\vec{t}_1
      &
        (\vec{t}_1+\vec{t}_2)+\vec{t}_3&\equiv&\vec{t}_1+(\vec{t}_2+\vec{t}_3)
      &
        1\cdot\vec{t}&\equiv&\vec{t}
      \\[6pt]
      \alpha\cdot(\beta\cdot\vec{t})&\equiv& \alpha\beta\cdot\vec{t}
      &
        (\alpha+\beta)\cdot\vec{t} &\equiv& \alpha\cdot\vec{t}+\beta\cdot\vec{t}
      &
        \alpha\cdot(\vec{t}_1+\vec{t}_2)&\equiv& \alpha\cdot\vec{t}_1+\alpha\cdot\vec{t}_2
    \end{array}
  \]
\caption{Congruence rules on term distributions}%
\label{tab:cong-term}
\end{table*}

From now on, we consider term distributions modulo the congruence~$\equiv$, and
simply write $\vec{t}=\vec{t}'$ for $\vec{t}\equiv\vec{t}'$. This convention
does not affect \emph{inner}---or \emph{raw}---distributions (which occur
within a pure term, for instance in the body of an abstraction), that are still
considered only up to $\alpha$-conversion\footnote{Intuitively, a distribution
  that appears in the body of an abstraction (or in the body of a
  let-construct, or in a branch of a match-construct) does not represent a real
  superposition, but \emph{machine code} that will produce later a particular
superposition, after some substitution has been performed.}.

Pure terms and term distributions are intended to be evaluated according to a
call-by-pure-values strategy\footnote{Called call-by-basis
  in~\cite{AssafDiazcaroPerdrixTassonValironLMCS14,DiazcaroGuillermoMiquelValironLICS19},
  and simply call-by-value in~\cite{ArrighiDowekRTA08,ArrighiDowekLMCS17}.},
  which is a declination of the call-by-value strategy in a computing
  environment where all functions are linear by construction. In its original
  form~\cite{ArrighiDowekRTA08,ArrighiDowekLMCS17}, a superposition $t(v+w)$
  reduced to $(tv+tw)$, while in our case
  following~\cite{DiazcaroGuillermoMiquelValironLICS19}, the first term is not
  even in the grammar, but it is just a notation for the former. This notation
  extends the syntactic constructs of the language by linearity, proceeding as
  follows: for all value distributions $\vec{v}=\sum_{i=1}^n\alpha_i\cdot{v_i}$
  and $\vec{w}=\sum_{j=1}^m\beta_j\cdot{w_j}$, and for all term distributions
  $\vec{s}_1, \vec{s}_2$, $\vec{t}=\sum_{k=1}^p\gamma_k\cdot t_k$, and
  $\vec{s}=\sum_{\ell=1}^q\delta_\ell \cdot s_\ell$ we have the notations given
  in Table~\ref{tab:notation}. Notice that $\vec t s$ is not in the grammar nor
  in the notation: the term at the left of an application must be a pure term.
  The reason is that we will focus in isometries, and the linear combination of
  isometries is not necessarily an isometry.

\begin{table*}[t]
  \[
    \begin{array}{r@{\ }l@{}r@{\ }l}
      \Pair{\vec{v}}{\vec{w}} &:= \sum\limits_{i=1}^n\sum\limits_{j=1}^k\alpha_i\beta_j\cdot\Pair{v_i}{w_j}
      & t\,\vec{s} &:= \sum\limits_{\ell=1}^q\delta_\ell\cdot ts_\ell
      \\
      \Inl{\vec{v}} &:=\sum\limits_{i=1}^n\alpha_i\cdot\Inl{v_i}
      & \vec{t};\vec{s} &:= \sum\limits_{k=1}^p\gamma_k\cdot(t_k;\vec{s})
      \\
      \Inr{\vec{v}} &:=\sum\limits_{i=1}^n\alpha_i\cdot\Inr{v_i}
      &\LetP{x}{y}{\vec{t}}{\vec{s}} &:= \sum\limits_{k=1}^p\gamma_k\cdot\big(\LetP{x}{y}{t_k}{\vec{s}}\big)
      \\
      \multicolumn{4}{l}{\Match{\vec{t}}{x_1}{\vec{s}_1}{x_2}{\vec{s}_2} :=}\\
      \multicolumn{4}{r}{\sum\limits_{k=1}^p\gamma_k\cdot \big(\Match{t_k}{x_1}{\vec{s}_1}{x_2}{\vec{s}_2}\big)}
  \end{array}
\]
  Where\quad
    $\vec{v} =\sum_{i=1}^n\alpha_i\cdot{v_i}$,
    \quad $\vec{w}=\sum_{j=1}^m\beta_j\cdot{w_j}$,
    \quad $\vec{t} =\sum_{k=1}^p\gamma_k\cdot t_k$,
    \quad $\vec{s}=\sum_{\ell=1}^q\delta_\ell \cdot s_\ell$
  \caption{Notations for linear constructions}%
  \label{tab:notation}
\end{table*}

We write $\Val$, $\ValD$, $\Lambda$, and $\vec{\Lambda}$ to the
sets of pure values, value distributions, pure terms, and term distributions respectively.

Finally, in Table~\ref{tab:AtomicEval} we give the rewrite relation. As usual, we
write $\lra^*$ for the reflexive and transitive closure of $\lra$.

\begin{table*}
  \begin{align*}
    (\Lam{x}{\vec{t}}\,)\,v&\lra \vec{t}\,[x:=v]\\
    \Void;\vec{s}&\lra \vec{s}\\
    \LetP{x}{y}{\Pair{v}{w}}{\vec{s}}&\lra \vec{s}[x:=v,y:=w]\\
    \Match{\Inl{v}}{x_1}{\vec{s}_1}{x_2}{\vec{s}_2}&\lra \vec{s}_1[x_1:=v]\\
    \Match{\Inr{v}}{x_1}{\vec{s}_1}{x_2}{\vec{s}_2}&\lra \vec{s}_2[x_2:=v]
  \end{align*}
  \[
    \begin{array}{c}
      \infer{s\,t\lra s\,\vec r}{t\lra \vec r}
      \qquad
      \infer{t\,v\lra r\,v}{t\lra r}
      \qquad
      \infer{t;\vec{s}\lra \vec r;\vec{s}}{t\lra \vec r}
      \qquad
      \infer{\LetP{x}{y}{t}{\vec{s}}\lra \LetP{x}{y}{\vec r}{\vec{s}}}{t\lra \vec r} \\[5pt]
      \infer{\Match{t}{x_1}{\vec{s}_1}{x_2}{\vec{s}_2}\lra \Match{\vec r}{x_1}{\vec{s}_1}{x_2}{\vec{s}_2}}{t\lra \vec r} \\[5pt]
      \infer{\alpha\cdot t+\vec s\lra\alpha\cdot\vec r+\vec s}{t\lra\vec r}
    \end{array}
  \]
  \caption{Rewrite rules}%
  \label{tab:AtomicEval}\vspace{-6pt}
\end{table*}

As the reader may have noticed from the grammar in Table~\ref{tab:Syntax} and
the congruence rules from Table~\ref{tab:cong-term}, there is no such a thing as
a ``null vector'' in the grammar, and so $0\cdot\vec{t}$ does not simplify.
Indeed, we do not want in general a null vector $\vec 0$ within term
distributions, since $0\cdot\vec v$ must have a type compatible with $\vec v$, while $\vec 0$ do not have any restriction. In fact, the set of value distributions do not form
a vector space for this reason. However, we still can define an analogous to an
inner product and, from it, a notion of orthogonality.

\begin{defi}[Pseudo inner product and orthogonality]\label{def:pseudoinnerproduct}
  Let $\vec v=\sum_{i=1}^n\alpha_i\cdot v_i$ and $\vec{w}=\sum_{j=1}^m\beta_j\cdot w_j$ be two value distributions in canonical form.
  Then we define the pseudo inner product
  $(\cdot\mid\cdot):\ValD\times\ValD\to\mathbb C$ as
  \[
    (\vec v\mid\vec w) :=
    \sum_{i=1}^n\sum_{j=1}^m\bar\alpha_i\beta_j\delta_{v_i,w_j}
  \]
  where $\delta_{v_i,w_j}$ is the Kronecker delta, i.e.~it is $1$ if $v_i=w_j$,
  and $0$ otherwise.
  We write $\vec v\perp\vec w$ if $(\vec v\mid\vec w)=0$.
\end{defi}

Remark that in this section, we are not defining the mathematical structure we
have, just pinpointing the fact that it is not a vector space, and so we cannot
define an inner product. However, we have defined a function, which we call
pseudo inner product, which is enough for the syntactic treatment of the
calculus we are introducing. In Section~\ref{sec:DenSem} we will give the
rigorous mathematical definitions needed to give a denotational semantics of
such a calculus.

\subsection{Types}
Types are produced by the following grammar
\[
  A:=
  \mathbb U\mid
  \sharp A\mid
  A+A\mid
  A\times A\mid
  A\rightarrow A
\]
The type $\sharp A$ is meant to type term distributions of pure terms of type $A$.
This is a subset of the grammar from
in~\cite{DiazcaroGuillermoMiquelValironLICS19}.
In particular, we do not include
the construction $\flat A$, however we use the notation $A^\flat$ (read: $A$ is flat) for the
following property: $A$ does not contain any $\sharp$, except, maybe, at the right of an arrow.
We also write $A\oplus
B:=\sharp(A+B)$ and $A\otimes B:=\sharp(A\times B)$.
In~\cite{DiazcaroGuillermoMiquelValironLICS19} there is also a type
$A\Rightarrow B$, which contains the superposition of arrows, which are valid
arrows. In our case, we decided to simplify the language by not allowing
superpositions of arrows to be arrows, and so this particular type construct is not used.

In Table~\ref{tab:subtyping} we give a subtyping relation between types. In particular, $\sharp\sharp A<\sharp A$ since a term distribution of term distributions of type $\sharp A$ is just a term distribution of pure terms of type $A$,

In Table~\ref{tab:types} we give the typing rules, where
we use the notation
  $\Gamma\vdash(\Delta_1\vdash \vec v_1\perp\Delta_2\vdash\vec v_2):A$ for
\[
  \left\{
    \begin{array}{l}
      \Gamma,\Delta_1\vdash\vec v_1:A\\
      \Gamma,\Delta_2\vdash\vec v_2:A\\
      \theta_{\Gamma,\Delta_1}(\vec v_1)\perp \theta_{\Gamma,\Delta_2}(\vec v_2)
    \end{array}
  \right.
\]
where for any context $\Gamma$, $\theta_\Gamma$ is a substitution of variables
by pure values of the same type. Notice that substituting a variable in a value by a pure value, keep the term being a value.
When $\Delta_1=\Delta_2=\emptyset$, we just write $\Gamma\vdash(\vec v_1\perp\vec v_2):A$.

The given type system is linear on types, except flat types (i.e.~any type $A$
such that $A^\flat$). Notice that the type system uses the notations from
Table~\ref{tab:notation}, for example, the rule
\[
  {\vcenter{\infer[^{\mathsf{Pair}}]{\Gamma,\Delta\vdash(\vec v,\vec w):A\times B}{\Gamma\vdash\vec v:A & \Delta\vdash\vec w:B}}}
\]
is in fact
\[
  {\vcenter{\infer[^{\mathsf{Pair}}]{\Gamma,\Delta\vdash\sum_{ij}\alpha_i\beta_j\cdot(v_i,w_j):A\times B}{\Gamma\vdash\sum_i\alpha_i\cdot v_i:A & \Delta\vdash\sum_j\beta_j\cdot w_j:B}}}
\]

\begin{table*}
  \[
    \begin{array}{c@{\qquad\qquad}c@{\qquad\qquad}c}
      {\infer{A\leq A}{}}
      &
      {\infer{A\leq C}{A\leq B & B\leq C}}
      &
      {\infer{A\leq\sharp A}{}}
      \qquad\qquad
      {\infer{\sharp\sharp A\leq\sharp A}{}}
      \\[1ex]
      {\infer{A'\rightarrow B\leq A\rightarrow B'}{A\leq A' & B\leq B'}}
      &
      \infer{A\times B\leq A'\times B'}{A\leq A' & B\leq B'}
      &
      \infer{A+B\leq A'+B'}{A\leq A' & B\leq B'}
    \end{array}
  \]
  \caption{Subtyping}%
  \label{tab:subtyping}
\end{table*}

\begin{table*}
  \[
    \begin{array}{c}
      {\vcenter{\infer[^{\mathsf{Ax}}]{x:A\vdash x:A}{}}}
      \qquad
      {\vcenter{\infer[^{\mathsf{Lam}}] {\Gamma\vdash\lambda x.\vec t:A\rightarrow B} {\Gamma,x:A\vdash\vec t:B}}}
      \qquad
      {\vcenter{\infer[^{\mathsf{App}}] {\Gamma,\Delta\vdash t\vec s:B} {\Gamma\vdash t:A\rightarrow B & \Delta\vdash\vec s:A}}}
      \\[10pt]
      {\vcenter{\infer[^{\mathsf{Void}}]{\vdash\Void:\mathbb U}{}}}
      \qquad
      {\vcenter{\infer[^{\mathsf{PureSeq}}]{\Gamma,\Delta\vdash t;\vec s:A}{\Gamma\vdash t:\mathbb U & \Delta\vdash\vec s:A}}}
      \qquad
      {\vcenter{\infer[^{\mathsf{UnitarySeq}}]{\Gamma,\Delta\vdash\vec t;\vec s:\sharp A}{\Gamma\vdash\vec t:\sharp \mathbb U & \Delta\vdash\vec s:\sharp A}}}
      \\[10pt]
      {\vcenter{\infer[^{\mathsf{Pair}}]{\Gamma,\Delta\vdash(\vec v,\vec w):A\times B}{\Gamma\vdash\vec v:A & \Delta\vdash\vec w:B}}}
      \\[10pt]
      {\vcenter{\infer[^{\mathsf{PureLet}}]{\Gamma,\Delta\vdash\mathsf{let}\ (x,y)=t\ \mathsf{in}\ \vec s:C}{\Gamma\vdash t:A\times B & \Delta,x:A,y:B\vdash\vec s:C}}}
      \\[10pt]
      {\vcenter{\infer[^{\mathsf{UnitaryLet}}]{\Gamma,\Delta\vdash\mathsf{let}\ (x,y)=\vec t\ \mathsf{in}\ \vec s:\sharp C}{\Gamma\vdash\vec t:A\otimes B & \Delta,x:\sharp A,y:\sharp B\vdash\vec s:\sharp C}}}
      \\[10pt]
      {\vcenter{\infer[^{\mathsf{InL}}]{\Gamma\vdash\mathsf{inl}(v):A+B}{\Gamma\vdash v:A}}}
      \qquad
      {\vcenter{\infer[^{\mathsf{InR}}]{\Gamma\vdash\mathsf{inr}(v):A+B}{\Gamma\vdash v:B}}}
      \\[10pt]
      {\vcenter{\infer[^{\mathsf{PureMatch}}]{\Gamma,\Delta\vdash\mathsf{match}\ t\ \{\mathsf{inl}(x_1)\mapsto\vec v_1\mid\mathsf{inr}(x_2)\mapsto\vec v_2\}:C}{\Gamma\vdash t:A+B & \Delta\vdash(x_1: A\vdash\vec v_1\perp x_2: B\vdash\vec v_2): C}}}
      \\[10pt]
      {\vcenter{\infer[^{\mathsf{UnitaryMatch}}]{\Gamma,\Delta\vdash\mathsf{match}\ \vec t\ \{\mathsf{inl}(x_1)\mapsto\vec v_1\mid\mathsf{inr}(x_2)\mapsto\vec v_2\}:\sharp C}{\Gamma\vdash\vec t:A\oplus B & \Delta\vdash(x_1:\sharp A\vdash\vec v_1\perp x_2:\sharp B\vdash\vec v_2):\sharp C}}}
      \\[10pt]
      \vcenter{\infer[^{\mathsf{Sup}}]{\vdash\sum_{j=1}^m\alpha_j\cdot\vec v_j:\sharp A}{\text{\scriptsize $(k\neq h)$} & \vdash(\vec v_k\perp\vec v_h):A & \sum_{j=1}^m|\alpha_j|^2=1 & m\geq 1 & A\neq B\rightarrow C }}
      \\[10pt]
      {\vcenter{\infer[^{\leq}]{\Gamma\vdash\vec t:B}{\Gamma\vdash\vec t:A & A\leq B}}}
      \qquad
      \vcenter{\infer[^{\equiv}]{\Gamma\vdash\vec r:A}{\Gamma\vdash\vec t:A & \vec t\equiv\vec r}}
      \\[10pt]
      {\vcenter{\infer[^{\mathsf{Weak}}]{\Gamma,x:A\vdash\vec t:B}{\Gamma\vdash\vec t:B & A^\flat}}}
      \qquad
      {\vcenter{\infer[^{\mathsf{Contr}}]{\Gamma,x:A\vdash\vec t[y:=x]:B}{\Gamma,x:A,y:A\vdash\vec t:B & A^\flat}}}
    \end{array}
  \]
  \caption{Typing system}%
  \label{tab:types}\vspace{-6pt}
\end{table*}

\section{Syntactic properties}\label{sec:ExpCorr}
This section is devoted to proving several syntactic properties of the calculus
introduced in the previous section. In particular, Progress
(Section~\ref{sec:progress}), Subject Reduction (Section~\ref{sec:SR}), and
Strong Normalization (Section~\ref{sec:SN}). The last property is shown by
proving that the calculus is a valid fragment with respect to the realizability
semantics given in~\cite{DiazcaroGuillermoMiquelValironLICS19}. Since \OC is a
fragment of a bigger calculus, we also show that it is expressive enough for quantum
computing (Section~\ref{sec:expressivity}).

\subsection{Progress}\label{sec:progress}
\begin{thm}[Progress]\label{thm:progress}
  If $\vdash \vec t:A$ and $\vec t$ does not reduce, then $\vec t\in\ValD$.
\end{thm}
\begin{proof}
  We proceed by induction on $\vec t$.
  \begin{itemize}
  \item If $\vec t$ is a value distribution, we are done.
  \item Let $\vec t=s\vec r$. Then $\vdash s:B\Rightarrow A$, but since $\vec t$
    does not reduce, neither does $s$, so, by the induction hypothesis $s\in\ValD$,
    and, due to its type, the only possibility is $s\equiv\lambda x.\vec s'$, which is
    absurd since $\vec t$ does not reduce.
  \item $\vec t=\vec s;\vec r$. Then there are two possibilities:
    \begin{itemize}
    \item $\vdash\vec s:\U$, but since $\vec t$ does not reduce, neither does
      $\vec s$, so, by the induction hypothesis $\vec s$ is a value, and, due its
      type, the only possibility is $\vec s\equiv\Void$, which is absurd since $\vec t$ does
      not reduce.
    \item $\vdash\vec s:\sharp\U$, but since $\vec t$ does not reduce, neither
      does $\vec s$, so, by the induction hypothesis $\vec s$ is a value, and, due its
      type, the only possibility is $\vec s\equiv\alpha\cdot\Void$, which is absurd since
      $\vec t$ does not reduce.
    \end{itemize}
  \item $\vec t=\LetP{x}{y}{\vec{s}}{\vec{r}}$. Then there are two
    possibilities:
    \begin{itemize}
    \item $\vdash\vec s:B\times C$, but since $\vec t$ does not reduce, neither
      does $\vec s$, so, by the induction hypothesis $\vec s$ is a value, and, due its
      type, the only possibility is $\vec s\equiv(v_1,v_2)$, which is absurd since $\vec t$
      does not reduce.
    \item $\vdash\vec s:\sharp(B\times C)$, but since $\vec t$ does not reduce,
      neither does $\vec s$, so, by the induction hypothesis $\vec s$ is a value, and,
      due its type, the only possibility is $\vec s\equiv(\vec v_1,\vec v_2)$, which is
      absurd since $\vec t$ does not reduce.
    \end{itemize}
  \item $\vec t=\Match{\vec s}{x_1}{\vec v_1}{x_2}{\vec v_2}$. Then there are
    two possibilities:
    \begin{itemize}
    \item $\vdash\vec s:B+C$, but since $\vec t$ does not reduce, neither does
      $\vec s$, so, by the induction hypothesis $\vec s$ is a value, and, due its
      type, the only possibilities are $\vec s\equiv\Inl{s'}$ or $\vec s\equiv\Inr{s'}$, both of
      which are absurd since $\vec t$ does not reduce.
    \item $\vdash\vec s:\sharp(B+C)$, but since $\vec t$ does not reduce,
      neither does $\vec s$, so, by the induction hypothesis $\vec s$ is a value, and,
      due its type, the only possibilities are $\vec s\equiv\Inl{\vec s'}$ or $\vec s\equiv\Inr{\vec s'}$, both of which are absurd since $\vec t$ does not reduce.
      \qedhere
    \end{itemize}
  \end{itemize}
\end{proof}

\subsection{Subject reduction}\label{sec:SR}
The type preservation with our chosen typing rules is proven now. We first need a substitution lemma.
\begin{lem}%
  \label{lem:SubstSR}
  Let $\Gamma,x:A\vdash\vec t:B$, $\Delta\vdash\vec v:A$, and $\Delta^\flat$, then $\Gamma,\Delta\vdash\vec t[x:=\vec v]:B$.
\end{lem}
\begin{proof}
  By induction on $\vec t$.
  \begin{itemize}
  \item If $x\notin\FV(\vec t)$ then a
    straightforward generation lemma shows that $x:A$ can be
    removed from $\Gamma,x:A\vdash\vec t:B$, and from $\Gamma\vdash\vec t:B$ we
    can derive $\Gamma,\Delta\vdash\vec t:B$ by rule $\s{Weak}$. Notice that
    $t[x:=\vec v]=\vec t$.
  \item Let $\vec t=x$, then $\Gamma^\flat$ and $A\leq B$, and so, by
    rules $\leq$ and $\mathsf{Weak}$, we have $\Gamma,\Delta\vdash\vec v:B$. Notice that
    $x[x:=\vec v]=\vec v$.
  \item Let $\vec t=\lambda y.\vec s$, then $C\Rightarrow D\leq B$ and
    $\Gamma,x:A,y:C\vdash\vec s:D$. Hence, by the induction hypothesis,
    $\Gamma,\Delta,y:C\vdash\vec s[x:=\vec v]:D$. Therefore, by rules $\s{Lam}$
    and $\leq$,
    $\Gamma,\Delta\vdash\lambda y.\vec s[x:=\vec v]:B$. Notice that $\lambda y.\vec s[x:=\vec v]=(\lambda y.\vec s)[x:=\vec v]$.
  \item Let $\vec t=(v_1,v_2)$, then $\Gamma_1\vdash v_1:B_1$ and
    $\Gamma_2\vdash v_2:B_2$ with $(\Gamma_1,\Gamma_2)=(\Gamma,x:A)$ and $B_1\times
    B_2=B$. Assume that $\Gamma_1=(\Gamma_1',x:A)$. Then, by the induction
    hypothesis, $\Gamma_1',\Delta\vdash v_1[x:=\vec v]:B_1$, and so, by rule
    $\s{Pair}$, $\Gamma_1',\Gamma_2,\Delta\vdash (v_1[x:=\vec v],v_2):B$. Notice that
    $(\Gamma_1',\Gamma_2)=\Gamma$ and $(v_1[x:=\vec v],v_2)=(v_1,v_2)[x:=\vec v]$.
    The case where $\Gamma_2=(\Gamma_2',x:A)$ is analogous.
  \item Let $\vec t=\Inl w$, then $\Gamma,x:A\vdash w:B_1$, with $B=B_1+B_2$,
    so, by the induction hypothesis, $\Gamma,\Delta\vdash w[x:=\vec v]:B_1$ and, by rule
    $\s{InL}$, $\Gamma,\Delta\vdash\Inl{w[x:=\vec v]}:B$. Notice that $\Inl{w[x:=\vec v]}=\Inl
    w[x:=\vec v]$.
  \item Let $\vec t=\Inr w$, this case is analogous to the previous.
  \item Let $\vec t=s\vec r$, then
    $\Gamma_1\vdash s:C\Rightarrow D$ and $\Gamma_2\vdash\vec r:C$,
      with $D\leq B$ and $(\Gamma_1,\Gamma_2)=(\Gamma,x:A)$,
    Assume $\Gamma_1=(\Gamma_1',x:A)$. Then, by the
    induction hypothesis, $\Gamma_1',\Delta\vdash s[x:=\vec v]:C\Rightarrow D$. Thus, by
    rules $\s{App}$, and $\leq$,
    $\Gamma_1',\Gamma_2,\Delta\vdash s[x:=\vec v]\vec r:B$. Notice that
    $(\Gamma_1',\Gamma_2)=\Gamma$ and $s[x:=\vec v]\vec r=(s\vec r)[x:=\vec v]$. The case where
    $\Gamma_2=(\Gamma_2',x:A)$ is analogous.
  \item Let $\vec t=\vec r;\vec s$, then there are two possibilities:
    \begin{enumerate}
    \item Either $\Gamma_1\vdash \vec r:\U$ and $\Gamma_2\vdash s:B$, with
      $(\Gamma_1,\Gamma_2)=(\Gamma,x:A)$,
    \item or $\Gamma_1\vdash \vec r:\sharp\U$ and $\Gamma_2\vdash\vec s:\sharp C$, with
      $B=\sharp C$, and $(\Gamma_1,\Gamma_2)=(\Gamma,x:A)$.
    \end{enumerate} In any case, assume $\Gamma_1=(\Gamma_1',x:A)$. Then, by the
    induction hypothesis, $\Gamma_1',\Delta\vdash\vec r[x:=\vec v]:E$ (with $E=\U$ in the
    first case or $E=\sharp\U$ in the second). Thus, by rules $\s{PureSeq}$ or
    $\s{UnitarySeq}$, $\Gamma_1',\Gamma_2,\Delta\vdash\vec r[x:=\vec v];\vec s:B$. Notice that
    $(\Gamma_1',\Gamma_2)=\Gamma$ and $\vec r[x:=\vec v];\vec s=(\vec r;\vec s)[x:=\vec v]$. The case where
    $\Gamma_2=(\Gamma_2',x:A)$ is analogous.
  \item Let $\vec t=\LetP{x_1}{x_2}{\vec r}{\vec s}$, then there are two possibilities:
    \begin{enumerate}
    \item Either $\Gamma_1\vdash\vec r:C\times D$ and
      $\Gamma_2,x_1:C,x_2:D\vdash\vec s:B$, with $(\Gamma_1,\Gamma_2)=(\Gamma,x:A)$,
    \item or $\Gamma_1\vdash\vec r:\sharp(C\times D)$ and $\Gamma_2,x_1:\sharp
      C,x_2:\sharp D\vdash\vec s:\sharp C$, with $B=\sharp C$ and
      $(\Gamma_1,\Gamma_2)=(\Gamma,x:A)$.
    \end{enumerate} In any case, assume $\Gamma_1=(\Gamma_1',x:A)$. Then, by the
    induction hypothesis, $\Gamma_1',\Delta\vdash\vec r[x:=\vec v]:E$ (with $E=C\times D$ in
    the first case or $E=\sharp(C\times D)$ in the second). Thus, by rules
    $\s{PureLet}$ or $\s{UnitaryLet}$, $\Gamma_1',\Gamma_2,\Delta\vdash \Gamma\vdash
    \LetP{x_1}{x_2}{\vec r[x:=\vec v]}{\vec s}:B$. Notice that $(\Gamma_1',\Gamma_2)=\Gamma$
    and $\LetP{x_1}{x_2}{\vec r[x:=\vec v]}{\vec s}=(\LetP{x_1}{x_2}{\vec r}{\vec s})[x:=\vec v]$. The
    case where $\Gamma_2=(\Gamma_2',x:A)$ is analogous.
  \item Let $\vec t=\Match{\vec r}{x_1}{\vec v_1}{x_2}{\vec v_2}$, then there
    are two possibilities:
    \begin{enumerate}
    \item Either $\Gamma_1\vdash\vec r:C+D$ and $\Gamma_2\vdash(x_1:C\vdash\vec v_1\perp x_2:D\vdash\vec v_2):B$, with $(\Gamma_1,\Gamma_2)=(\Gamma,x:A)$,
    \item or $\Gamma_1\vdash\vec r:\sharp(C+D)$ and $\Gamma_2\vdash(x_1:\sharp
      C\vdash\vec v_1\perp x_2:\sharp D\vdash\vec v_2):\sharp C$, with $B=\sharp C$
      and $(\Gamma_1,\Gamma_2)=(\Gamma,x:A)$.
    \end{enumerate}
    In any case, assume $\Gamma_1=(\Gamma_1',x:A)$. Then, by the induction
    hypothesis, $\Gamma_1',\Delta\vdash\vec r[x:=\vec v]:E$ (with $E=C+D$ in the
    first case or $E=\sharp(C+D)$ in the second). Thus, by rules $\s{PureMatch}$ or
    $\s{UnitaryMatch}$, $\Gamma_1',\Gamma_2,\Delta\vdash \Gamma\vdash \Match{\vec r[x:=\vec v]}{x_1}{\vec v_1}{x_2}{\vec v_2}:B$. Notice that
    $(\Gamma_1',\Gamma_2)=\Gamma$ and $\Match{\vec r[x:=\vec v]}{x_1}{\vec v_1}{x_2}{\vec v_2}=(\Match{\vec r}{x_1}{\vec v_1}{x_2}{\vec v_2})[x:=\vec v]$.

    Now assume $\Gamma_2=(\Gamma_2',x:A)$, then by the induction hypothesis,
    $\Gamma_2,\Delta,x_1:C\vdash\vec v_1[x:=\vec v]:B$ and
    $\Gamma_2,\Delta,x_2:D\vdash\vec v_2[x:=\vec v]:B$. Notice that the $\perp$
    condition is preserved under substitution, therefore,
    $\Gamma_2,\Delta\vdash(x_1:C\vdash\vec v_1[x:=\vec v]\perp x_2:D\vdash\vec v_2[x:=\vec v]):B$, and we close analogously to the previous case.
  \item Let $\vec t=\sum_{i=1}^n\alpha_i\cdot\vec v_i$, then $x\notin\FV(\vec t)$ and we are in the first case.
    \qedhere
  \end{itemize}
\end{proof}

\begin{thm}[Subject reduction]\label{thm:SR}
  If $\Gamma\vdash\vec t:A$ and $\vec t\lra\vec s$, then $\Gamma\vdash\vec s:A$.
\end{thm}
\begin{proof}
  By induction on the relation $\lra$. The proof is straightforward, using
  Lemma~\ref{lem:SubstSR}.
\end{proof}

\subsection{Strong normalization}\label{sec:SN}
We prove that \OC is valid with respect to the realizability model given
in~\cite{DiazcaroGuillermoMiquelValironLICS19} (Theorem~\ref{thm:correctness}),
which implies strong normalization (Corollary~\ref{cor:SN})\footnote{Notice,
  however, that subject reduction is not implied by the correctness with respect
  to the realizability semantics, since it may be the case that a term $\vec t$
  reduces to a term $\vec t'$, both in the semantics, but the second not typable
  with the typing rules chosen. This is why we have given a direct syntactic proof
  of such a property (Section~\ref{sec:SR}).}.

In~\cite{DiazcaroGuillermoMiquelValironLICS19} a realizability
semantics has been defined, and the type system of the language is determined by any rule
following such a semantics. In particular, the realizability
predicate~\cite[Def. IV.2]{DiazcaroGuillermoMiquelValironLICS19} states that
a term $\vec t$ is a realizer of a type $A$ (notation $\vec t\Vdash A$) if and only if $\vec t$ rewrites to a value in the interpretation of $A$.
Then, a typing judgement $\vdash\vec t:A$ is valid if and only if, $\vec t\Vdash
A$. In this paper, we have fixed a set of typing rules in Table~\ref{tab:types},
some of which are already proven to be valid
in~\cite{DiazcaroGuillermoMiquelValironLICS19}, while others are proved to be correct
next (Theorem~\ref{thm:correctness}).
For the sake of self-containment, we include the needed definitions from the
realizability semantics~\cite{DiazcaroGuillermoMiquelValironLICS19}.

Let $\mathcal{S}_1=\{\vec v:(\vec v|\vec v)=1\}$.
The interpretation $\semR \cdot$ of types is given by
\begin{align*}
  \semR{\U}&=\{\Void\}\\
  \semR{\sharp A} &=\mathsf{span}(\semR A)\cap\mathcal{S}_1\\
  \semR{A+B} &=\{\Inl{\vec v}:\vec v\in\semR{A}\}\cup\{\Inr{\vec w}:\vec w\in\semR B\}\\
  \semR{A\times B} &=\{(\vec v,\vec w):\vec v\in\semR A,\vec w\in\semR B\}\\
  \semR{A\rightarrow B} &=\{\lambda x.\vec t:\forall\vec v\in\semR A,\vec t\langle x:=\vec v \rangle\Vdash B\}
\end{align*}
where
$\vec t\Vdash A$ means that $\vec t$ reduces to a value in $\semR A$, and $\vec t\langle x:=\vec v \rangle$ is the bilinear substitution defined as follows:
Let $\vec t=\sum_i\alpha_i\cdot t_i$ and $\vec v=\sum_j\beta_j\cdot v_j$. Then,
$\vec t\langle  x:=\vec v \rangle:=\sum_i\sum_j\alpha_i\beta_j\cdot t_i[x:=v_j]$.

If $\sigma$ is a substitution, we may write $\vec t\langle \sigma \rangle$ for
the term distribution $\vec t$ substituted by $\sigma$.
In addition, we write $\sigma\in\semR\Gamma$ if for all $x:A\in\Gamma$,
$x\langle \sigma \rangle\Vdash A$.

Finally, the realizability semantics defines the typing rules as follows:
$\Gamma\vdash\vec t:A$ is a notation for $\forall\sigma\in\semR\Gamma$, $\vec t\langle \sigma \rangle\Vdash A$. Hence, we need to prove that the typing system
presented in Table~\ref{tab:types} is correct, which is done by Theorem~\ref{thm:correctness}.

\begin{thm}[Correctness]\label{thm:correctness}
  If $\Gamma\vdash\vec t:A$, then for all $\sigma\in\semR\Gamma$, we have $\vec t\<\sigma\>\Vdash A$.
\end{thm}
\begin{proof}
  We only prove the judgements that are not already proved
  in~\cite{DiazcaroGuillermoMiquelValironLICS19}. We proceed by induction on the
  typing derivation.
  \begin{itemize}
  \item Rule $\mathsf{Lam}$. By the induction hypothesis,
    $\forall(\sigma,x:=\vec w)\in\semR{\Gamma,x:A}$, $\vec t\<\sigma,x:=\vec w\>\Vdash B$.
    Therefore, by definition, $(\lambda x.\vec t)\<\sigma\>=\lambda x.\vec t\<\sigma\>\in\semR{A\rightarrow B}$. So, $(\lambda x.\vec t)\<\sigma\>\Vdash
    A\rightarrow B$.

  \item Rule $\mathsf{App}$. By the induction hypothesis,
    $\forall\sigma\in\semR\Gamma$, $t\<\sigma\>\Vdash A\rightarrow B$ and
    $\forall\theta\in\semR\Delta$, $\vec s\<\theta\>\Vdash A$. Then, by definition,
    $t\<\sigma\>\lra^*\lambda x.\vec r$ and $\vec s\<\theta\>\lra^*\vec v\in\semR A$
    such that $\vec r\<x:=\vec v\>\Vdash B$. Since $(t\vec s)\<\sigma,\theta\>\lra^*\vec r\<x:=\vec v\>$, we have, $(t\vec s)\<\sigma,\theta\>\Vdash B$.

  \item Rule $\mathsf{Sup}$. For all $j$, by the induction hypothesis, we have
    that $\vec v_j\in\semR A$. Since for all $k\neq h$ we have $\vec v_k\perp\vec v_h$, and $\sum_j|\alpha_j|^2=1$, we have $\sum_j\alpha_j\cdot\vec v_j\in\So$,
    and so $\sum_j\alpha_j\cdot\vec v_j\in\semR{\sharp A}$. Hence,
    $\sum_j\alpha_j\cdot\vec v_j\Vdash\sharp A$.
    \qedhere
  \end{itemize}
\end{proof}

The following corollary is a direct consequence of Theorems~\ref{thm:progress}
and~\ref{thm:correctness}.

\begin{cor}[Strong normalization]\label{cor:SN}
  If $\Gamma\vdash\vec t:A$ then $\vec t$ is strongly normalizing.
  \qed
\end{cor}

The realizability model also implies that $\ValD\subseteq\mathcal{S}_1$.

\subsection{Expressivity}\label{sec:expressivity}
First we define the following encoding of norm-$1$ vectors from $\mathbb C^{2^n}$ to
terms in \OC: % chktex 13
\begin{defi}\label{def:encoding}
  Let $b_i\in\{0,1\}$ and $\ket{\underline{k}}=\ket{b_0\dots b_{2^n-1}}$, where
  $b_0\dots b_{2^n-1}$ is the binary representation of $k\in\mathbb N$. Then, we
  encode $\ket{\underline{k}}$ in \OC as $\hat{{\underline k}} =
  (\hat{b}_0,(\hat{b}_2,(\dots,\hat{b}_{2^n-1})))$, with $\hat{0}=\Inl\Void$ and
  $\hat{1}=\Inr\Void$.
  Thus, if $\sum_{i=0}^{2^n-1}|\alpha_i|^2=1$ and $\vec{\mathsf v}=(\alpha_0,\dots,\alpha_{2^n-1})^T\in\mathbb
  C^{2^{n}}$, its encoding in \OC is \( \hat{\vec{\mathsf v}}=\sum_{k=0}^{2^n-1}
  \alpha_k\cdot\hat{\underline{k}} \).
\end{defi}

From now on, we may write $\mathbb B$ for $\U+\U$, $A^n$ for
$\prod_{i=1}^nA=A\times A\times\cdots\times A$, and $A^{\otimes n}$ for $\sharp
A^n=\sharp(A\times A\times\cdots\times A)$.

\begin{exa}
  $(\frac 1{\sqrt 2},0,0,\frac 1{\sqrt 2},0,0,0,0)^T=\frac 1{\sqrt
    2}\ket{000}+\frac 1{\sqrt 2}\ket{011}=\frac 1{\sqrt 2}\ket{\underline
    0}+\frac 1{\sqrt 2}\ket{\underline 3}$ in $\C^{2^3}$ is encoded as
  \[
    \frac 1{\sqrt 2}\cdot(\Inl\Void,(\Inl\Void,\Inl\Void))+ \frac 1{\sqrt
      2}\cdot(\Inl\Void,(\Inr\Void,\Inr\Void))
  \]
\end{exa}

The previous construction is typable as shown in the following example.
\begin{exa}\label{ex:type}
  Let $\vec{\s v}=(\alpha_0,\alpha_1,\dots,\alpha_{7})^T=\sum_{k=0}^7\alpha_{k}\ket{\underline{k}}\in\mathbb
  C^{2^3}$, with $\sum_{k=0}^{7}|\alpha_k|^2=1$. Hence, $\hat{\vec{\s v}}=\sum_{k=0}^{7}\alpha_{k}\cdot\hat{\underline{k}}$, using the encoding
  from Definition~\ref{def:encoding}.

  We check that $\vdash\hat{\vec{\s v}}:\mathbb B\otimes\mathbb B\otimes\mathbb B$. We have
    \[
      \infer[\s{Pair}]{\vdash \hat{\underline{0}}:\mathbb B\times\mathbb B\times\mathbb B}{
        \infer[\s{InL}]{\vdash\Inl\Void:\mathbb B}{\infer[\s{Void}]{\vdash\Void:\U}{}}
        &
        \infer[\s{Pair}]{\vdash(\Inl\Void,\Inl\Void):\mathbb B\times\mathbb B}
        {
          \infer[\s{InL}]{\vdash\Inl\Void:\mathbb B}{\infer[\s{Void}]{\vdash\Void:\U}{}}
          &
          \infer[\s{InL}]{\vdash\Inl\Void:\mathbb B}{\infer[\s{Void}]{\vdash\Void:\U}{}}
        }
      }
    \]
    Similarly, we derive $\vdash\hat{\underline{1}}:\mathbb B\times\mathbb B\times\mathbb B$,\dots,
    $\vdash\hat{\underline{7}}:\mathbb B\times\mathbb B\times\mathbb B$. Therefore,
    \[
      \infer[\s{Sup}]{\vdash\hat{\vec{\s v}}:\mathbb B\otimes\mathbb B\otimes\mathbb B}
      {
        \text{\scriptsize $(j\neq k)$}
        &
        \vdash(\hat{\underline j}\perp\hat{\underline k}):\mathbb B\times\mathbb B\times\mathbb B
        &
        \sum_{i=0}^{7}|\alpha_i|^2=1
      }
    \]
\end{exa}

We can define a case construction for the elements of $\mathbb B^{\otimes n}$,
noted as $\s{case}\ s\ \s{of}\ \{\hat{\underline k}\mapsto\vec v_k\}$, as shown
in the following example.
\begin{exa}\label{ex:caseof}
  Let $\lambda z.\s{case}\ z\ \s{of}\ \{ \hat{\underline 0}\mapsto \vec v_0\mid
  \hat{\underline 1}\mapsto \vec v_1\mid \hat{\underline 2}\mapsto \vec v_2\mid
  \hat{\underline 3}\mapsto \vec v_3 \}$ be defined as
  \begin{align*}
    \lambda z. \LetP x y z \matchkeyword~x~
        \{ &\Inl{{x'}}\mapsto {x'};\matchkeyword~y~\{ \Inl{y'}\mapsto {{y'};\vec v_0}\mid
                                                    \Inr{y'}\mapsto {{y'};\vec v_1}\}\\
        \mid &\Inr{{x'}}\mapsto {x'};\matchkeyword~y~\{ \Inl{y'}\mapsto {{y'};\vec v_2}\mid
                                                    \Inr{y'}\mapsto {{y'};\vec v_3}\}\}
  \end{align*}

  Assuming $\vec v_i$ are orthogonal two by two, this term can be typed with
  $(\mathbb B\otimes\mathbb B)\rightarrow(\mathbb B\otimes\mathbb B)$,
  Extending this construction to any $n$ is an easy exercise.
\end{exa}

\begin{exa}
  We can use the case construction from Example~\ref{ex:caseof} to construct any
  quantum operator. For example, the CNOT operator corresponds to the matrix
  \[
    \left( \begin{smallmatrix}
        1 & 0 & 0 & 0\\
        0 & 1 & 0 & 0\\
        0 & 0 & 0 & 1\\
        0 & 0 & 1 & 0
      \end{smallmatrix}\right)
  \]
  which sends $\ket{0x}$ to $\ket{0x}$, $\ket{10}$ to $\ket {11}$, and
  $\ket{11}$ to $\ket{10}$. This operator can be written in \OC as
  \[
    \lambda z.\s{case}\ z\ \s{of}\ \{ \hat{\underline 0}\mapsto \hat{\underline 0}\mid \hat{\underline 1}\mapsto \hat{\underline 1}\mid \hat{\underline 2}\mapsto \hat{\underline 3}\mid \hat{\underline 3}\mapsto \hat{\underline 2} \}
  \]

  Any quantum operator can be written as a matrix, and this encoding provides exactly that.
\end{exa}

The expressivity of the language is stated as follows.
\begin{thm}[Expressivity]\label{thm:expressiveness}\leavevmode
  \begin{enumerate}
  \item Simply Typed Lambda Calculus extended with pairs and sums, in
    call-by-value, is included in \OC. % chktex 13
  \item If $U$ is an isometry acting on $\mathbb C^{2^n}$, then there exists a
    lambda term $\hat{U}$ such that $\vdash\hat{U}:\mathbb B^{\otimes n}\rightarrow
    \mathbb B^{\otimes n}$ and for all $\vec{\s v}\in\mathbb C^{2^n}$, if $\vec{\s w}=U\vec{\s v}$, we have $\hat U\hat{\vec{\s v}}\lra^*\hat{\vec{\s w}}$.
  \end{enumerate}
\end{thm}
\begin{proof}\leavevmode
  \begin{enumerate}
  \item The terms from the lambda calculus with pairs and sums are included in
    the grammar of pure terms. The set of types $\{A\mid A^\flat\}$ includes the
    simply types. Simply types, together with rules $\mathsf{Ax}$, $\mathsf{Lam}$,
    $\mathsf{App}$, $\mathsf{Pair}$, $\mathsf{PureLet}$, $\mathsf{InL}$,
    $\mathsf{InR}$, $\mathsf{PureMatch}$, $\mathsf{Weak}$, and $\mathsf{Cont}$ allow
    to type the lambda calculus extended with pairs and sums.
  \item Let $U$ be an isometry. We can define it by giving its behaviour on a
    base of $\C^{2^n}$. So, consider the canonical base $B=\{\ket{\underline
      0},\dots,\ket{\underline{2^n-1}}\}$ and for all $\ket{\underline k}\in B$, let
    $U\ket{\underline{k}}=(\beta_{0,k},\dots,\beta_{2^n-1,k})^T$. That is,
    $U=(\beta_{ij})_{ij}$. We can define a term $\hat U$ using the case construction
    introduced before (cf.~Example~\ref{ex:caseof}).
    \[
      \hat U = \lambda x.\s{case}\ x\ \s{of}\ \{
      \hat{\underline{k}}\mapsto\sum_{i=0}^{2^n-1}\beta_{ik}\cdot\hat{\underline{i}}\}
    \]
    We have $\vdash\hat U:\mathbb B^{\otimes
      n}\rightarrow\mathbb B^{\otimes n}$ (cf.~Example~\ref{ex:caseof}).

    Let $\vec{\s v}=(\alpha_0,\dots,\alpha_{2^n-1})^T=\sum_{k=0}^{2^n-1}\alpha_{k}\ket{\underline{k}}\in\mathbb
    C^{2^{n}}$, with $\sum_{i=0}^{2^n-1}|\alpha_i|^2=1$. Then, we have $\hat{\vec{\s v}}=\sum_{k=0}^{2^n-1}\alpha_k\cdot\hat{\underline k}$, with $\vdash\hat{\vec{\s v}}:\mathbb B^{\otimes n}$ (cf.~Example~\ref{ex:type}), therefore, $\vdash\hat U\hat{\vec{\s v}}:\mathbb B^{\otimes n}$.

    Notice that,
    \begin{align*}
      \hat U\hat{\vec{\s v}}
      &\equiv\sum_{k=0}^{2^n-1}\alpha_k\cdot \hat U\underline{\hat k}
      \lra^*\sum_{k=0}^{2^n-1}\alpha_k\cdot\sum_{i=0}^{2^n-1}\beta_{ik}\cdot\underline{\hat i}\\
      &\equiv\sum_{i=0}^{2^n-1}\sum_{k=0}^{2^n-1}\alpha_k\beta_{ik}\cdot\underline{\hat i}=\widehat{U\vec{\s v}}
    \qedhere
    \end{align*}
  \end{enumerate}
\end{proof}

\noindent
Since \OC is a fragment of the calculus that can be defined using the
realizability model, and the previous lemma shows that any isometry can be
represented in it, then the following theorem is still valid.
\begin{thm}[\teoremaunitarias]\label{thm:isometryLICS}
  A closed $\lambda$-abstraction $\lambda x.\vec t$ is a value of type
  $\sharp\mathbb B\rightarrow\sharp\mathbb B$ if an only if it represents an
  isometry $U:\C^2\rightarrow\C^2$.
  \qed
\end{thm}

Extending this result to a bigger dimension is straightforward.

\begin{rem}
  Even if we can check orthogonality on open terms,
  e.g.~$\Gamma\vdash(\Delta_1\vdash\vec v_1\perp\Delta_2\vdash\vec v_2):A$, we
  cannot type a constructor of an oracle (as, for example, the oracle needed for
  the Deutsch's algorithm, cf.~\cite[\S 1.4.3]{NC00}) parametrized by a given
  function $f$. That is, the oracle $U_f$ sending $\ket{b_1b_2}$ to
  $\ket{b_1,b_2\oplus f(b_1)}$ can be typed with $\mathbb B^{\otimes
  2}\rightarrow\mathbb B^{\otimes 2}$ for any given $f$ in our language,
  however, we cannot type a term $\lambda f.U_f$ such as
  \begin{align*}
    \lambda f.U_f:=\lambda f.\lambda x.\s{case}~x~\s{of}~\{
    &\hat{\underline 0}\mapsto(\Inl\ast,f\Inl\ast),\\
    &\hat{\underline 1}\mapsto(\Inl\ast,\s{not}~(f\Inl\ast)),\\
    &\hat{\underline 2}\mapsto(\Inr\ast,f\Inr\ast),\\
    &\hat{\underline 3}\mapsto(\Inr\ast,\s{not}~(f\Inr\ast))\}
  \end{align*}
  Indeed, the branches of the case are not values and so the orthogonality cannot be verified.

  In~\cite{DiazcaroGuillermoMiquelValironLICS19} this term $\lambda f.U_f$ is
  valid, with typing $\vdash \lambda f.U_f:(\mathbb B\rightarrow\mathbb
  B)\rightarrow \mathbb B^{\otimes 2}\rightarrow\mathbb B^{\otimes 2}$.
  Certainly, the orthogonality verification is done by the realizability model,
  by considering all the reductions of the application of the term to any
  possible argument, something that is not desirable in a static type system,
  where reducing a term in order to type it is not a good practice. In any
  case, the Deustch's algorithm, or any other quantum algorithm using such kind
  of oracle, does not construct the gate dynamically with a term of the kind
  $\lambda f.U_f$, but the unitary $U_f$ is given.
\end{rem}

\section{Denotational semantics}\label{sec:DenSem}
We define two categories, $\SetV$ and $\VecV$, and an adjunction between them.
Types are interpreted as objects in $\SetV$, and terms as maps using the
adjunction.

As stated in Section~\ref{sec:LambdaS1}, value distributions do not form a vector
space. Here we make this concept more precise by defining the non-standard notion of
distributive-action space\footnote{In~\cite{DiazcaroGuillermoMiquelValironLICS19}
  the notion of weak vector space is defined, which is based on a commutative
  monoid. Here we use a commutative semi-group instead, since
  a null vector is not needed.}.

\begin{defi}[Distributive-action space]\label{def:svs}
  A distributive-action space over a field $K$ is a commutative semi-group\footnote{That
  is, an associative and commutative magma.} $(V,+)$
  equipped with a scalar multiplication $(\cdot):K\times V\to V$ such that for
  all $\vec v,\vec w\in V$, $\alpha,\beta\in K$, we have the following axioms.
  \[
    \begin{array}{r@{\ }l@{\qquad}r@{\ }l}
      1\cdot \vec v &=\vec v
      &                    (\alpha+\beta)\cdot\vec v &=\alpha\cdot\vec v+\beta\cdot\vec v\\
      \alpha\cdot(\beta\cdot\vec v) &=\alpha\beta\cdot\vec v
      &    \alpha\cdot (\vec v+\vec w) &=\alpha\cdot\vec v+\alpha\cdot\vec w
    \end{array}
  \]
  In analogy with vector spaces, the elements of $V$ are called vectors and the
  elements of $K$ scalars.
\end{defi}
The notion of distributive-action space differs from the traditional notion of
vector space in that the underlying additive structure $(V,+)$ is an arbitrary
commutative semi-group, whose elements in general do not include a neutral
element, and so do not have an additive inverse. In a distributive-action space,
the vector $(-1)\cdot\vec v$ is in general not the additive inverse of $\vec v$,
and the product $0\cdot\vec v$ does not simplify to a neutral element $\vec 0$.
Indeed, term distributions do not have a null vector.

The notions of inner product and norm can be generalized to the case of
distributive-action spaces in the following way.

\begin{defi}[Inner product of a distributive-action space over $\mathbb C$]\label{def:innerproduct}
  An inner product of a distributive-action space $V$ is a function $\langle
  \cdot\mid\cdot \rangle:V\times V\to\mathbb C$ satisfying the following
  properties. For all $\alpha\in\mathbb C$, $\vec u,\vec v,\vec w\in V$,
  \begin{enumerate}
  \item $\langle \vec u\mid\vec v \rangle = \overline{\langle \vec v\mid\vec u\rangle}$, where $\overline\alpha$ is the conjugate of $\alpha$.
  \item $\langle\vec u\mid\alpha\cdot\vec v  \rangle =\alpha\langle \vec u\mid\vec v\rangle$

    $\langle \vec u\mid\vec v+\vec w \rangle = \langle \vec u\mid\vec v\rangle+\langle \vec u+\vec w \rangle$.

  \item $\langle \vec u\mid\vec u \rangle > 0$ for all $\vec u\in \{\vec v\in V\mid \forall\vec w, \vec v\neq 0.\vec w\}$.
  \end{enumerate}
\end{defi}
\begin{defi}[Norm of a distributive-action space]\label{def:norm}
  A norm of a distributive-action space $V$ is a function
  $||\cdot||:V\to\mathbb R^+$ satisfying the following properties. For all $\alpha\in K$, $\vec v,\vec w\in V$,
  \begin{enumerate}
  \item $||\vec v+\vec w||\leq||\vec v||+||\vec w||$.
  \item $||\alpha\cdot\vec v|| = |\alpha|||\vec v||$.
  \item $||\vec v||=0$ if and only if $\vec v=0.\vec w$, for some $\vec w$.
  \end{enumerate}
\end{defi}
\begin{thm}
  $\ValD$ is a normed distributive-action space over $\mathbb C$.
\end{thm}
\begin{proof}
  Verifying that $\ValD$ is a distributive-action space over $\mathbb C$ is
  straightforward by checking that the congruence given in
  Table~\ref{tab:cong-term} coincides with the requirements from
  Definition~\ref{def:svs}.

  We also check that Definition~\ref{def:pseudoinnerproduct} verifies the
  definition of an inner product (Definition~\ref{def:innerproduct}).

  Finally, we define the norm
  \[
    ||{\vec v}|| :=\sqrt{\langle \vec v\mid\vec v \rangle} = \sqrt{\sum_{i=1}^n|\alpha_i|^2}
  \]
  and check that such a norm verifies Definition~\ref{def:norm}.
\end{proof}
We write $\parts{\cdot} = \mathcal P(\cdot)\setminus\{\emptyset\}$, and
$\Definible$ to the set of (computable) functions that can be defined in \OC. % chktex 13

With all the previous definitions, we can define the categories which will give
the adjunction to model the calculus.

\begin{defi}[Category $\SetV$]
  The monoidal category $\SetV$ has the following elements:
  \begin{itemize}
  \item $\Ob\SetV=\parts{\ValD}$.
  \item For all $A,B\in\Ob\SetV$,
    \[ \HomS{A,B}=\{f\in\Definible\mid f:A\to B\in\Arr{\Set}\}
    \]
  \item For all $A,B\in\Ob\SetV$,
    \begin{align*} A\boxtimes B
      =&\{\textstyle\sum_i\sum_j\alpha_i\beta_j\cdot(v_i,w_j)\mid
         \begin{array}{l} \sum_i\alpha_i\cdot v_i\in A\\ \sum_j\beta_j\cdot
           w_j\in B
         \end{array}\}\\ \stackrel{\rm def}{=}&\{(\vec v,\vec w)\mid\vec v\in A\
                                                \text{and}\ \vec w\in B\}
    \end{align*}
  \item $1=\{\Void\}\in\Ob\SetV$.
  \item The obvious structural maps: $A\xlra{\lambda_\boxtimes}1\boxtimes A$,
    $A\xlra{\rho_\boxtimes}A\boxtimes 1$, and $(A_1\boxtimes A_2)\boxtimes
    A_3\xlra{\alpha_\boxtimes}A_1\boxtimes (A_2\boxtimes A_3)$.
  \end{itemize}
\end{defi}

\begin{rem}
  Let $A\xlra f B$ and $C\xlra g D$ in $\Arr\SetV$. Then $f,g\in\Definible$. So,
  let $\hat f$ and $\hat g$ be the terms in \OC implementing $f$ and $g$
  respectively. Therefore, a term implementing $f\boxtimes g$ is the following
  \( \lambda x.\LetP{y}{z}{x}{(\hat fy,\hat gz)} \).
\end{rem}

\begin{rem}
  Notice that while the category $\SetV$ is generated by the syntax of the calculus, it is not a syntactic category in the sense of~\cite[Part I $\S10$]{LambekScott}, where the objects are types and arrows are terms. Indeed, the objects at $\SetV$ are non-empty powersets of values, and the interpretation of the language will be done by using an adjunction with a richer category $\VecV$ yet to be define (cf.~Definition~\ref{def:VecV} and Proposition~\ref{prop:adjunction} stablishing the adjunction). Hence, this construction is far from trivial.
\end{rem}

In order to define the category $\VecV$, we define the set $\Hil$ and the map $\Span{}$.
\begin{defi}\label{def:functors}\leavevmode
  \begin{enumerate}
  \item $\Hil$ is the set of all the sub-distributive action spaces\footnote{$V$ is a
      sub-distributive action space of $W$ if $V\subseteq W$ and $V$ forms a distributive-action space
      with the operations from $W$ restricted to $V$.} of $\ValD$, with inner
    product induced by $\ValD$.
  \item $\Span{}:\parts{\ValD}\lra\Hil$ is a map defined by
    $A\mapsto\{\sum_i\alpha_i\cdot\vec v_i\mid\vec v_i\in A\}$.
  \item Let $A,B\in\parts{\ValD}$ and $f:A\lra B$ be a map. Then, $\Span f:\Span
    A\lra\Span B$ is the map defined by $\sum_i\alpha_i\cdot\vec v_i\mapsto\sum_i\alpha_i\cdot f\vec v_i$.
  \end{enumerate}
\end{defi}

\begin{defi}[Category $\VecV$]\label{def:VecV}
  The monoidal category $\VecV$ has the following elements:
  \begin{itemize}
  \item $\Ob\VecV = \Hil$.
  \item For all $V, W\in\Ob\VecV$,
    \[
      \HomV{V,W}=\{f\in\Definible\mid V\xlra{f} W\textrm{ is a linear map}\}
    \]
  \item For all $V,W\in\Ob\VecV$,
    \[
      V\otimes W=\Span(V\boxtimes W)
    \]
  \item $\I= \{\alpha\cdot\Void\mid\alpha\in\C\}\in\Ob\VecV$.
  \item The obvious structural maps:
    $V\xlra{\lambda_\otimes}I\otimes V$,
    $V\xlra{\rho_\otimes}V\otimes I$, and
    $(V_1\otimes V_2)\otimes V_3\xlra{\alpha_\otimes}V_1\otimes (V_2\otimes V_3)$.
  \end{itemize}
\end{defi}

\begin{lem}
  The following maps are monoidal functors:
  \begin{enumerate}
  \item $\Span{}:\SetV\lra\VecV$, defined as in Definition~\ref{def:functors}.
  \item $\Forg{}:\VecV\lra\SetV$, the forgetful functor.
  \end{enumerate}
\end{lem}
\begin{proof}\leavevmode
  \begin{enumerate}
  \item If $A\in\Ob\SetV$, then, by definition $\Span A\in\Ob\VecV$.

    Let $A\xlra f B\in\Arr\SetV$ and $\sum_{i=1}^n\alpha_i\cdot a_i,
    \sum_{j=1}^m\beta_j\cdot a'_j\in\Span A$. Then,
    \begin{align*}
      &\textstyle(\Span f)(\delta_1\cdot\sum_{i=1}^n\alpha_i\cdot a_i+\delta_2\cdot\sum_{j=1}^m\beta_j\cdot a'_j)\\
      &\textstyle=(\Span f)(\sum_{i=1}^n\delta_1\alpha_i\cdot a_i+\sum_{j=1}^m\delta_2\beta_j\cdot a'_j)\\
      &\textstyle=\sum_{i=1}^n\delta_1\alpha_i\cdot fa_i+\sum_{j=1}^m\delta_2\beta_j\cdot fa'_j\\
      &\textstyle=\delta_1\cdot\sum_{i=1}^n\alpha_i\cdot fa_i+\delta_2\cdot\sum_{j=1}^m\beta_j\cdot fa'_j\\
      &\textstyle=\delta_1\cdot(\Span f)(\sum_{i=1}^n\alpha_i\cdot a_i)+\delta_2\cdot(\Span f)(\sum_{j=1}^m\beta_j\cdot a'_j)
    \end{align*}
    Therefore, $\Span f:\Span A\rightarrow \Span B$ is lineal and since
    $f\in\Definible$, $\Span f\in\Definible$. Thus $\Span f\in\Arr\VecV$.
    Functoriality is fulfilled by being a span.
  \item If $V\in\Ob\VecV\subseteq\parts{\ValD}=\Ob\SetV$, then, $\Forg
    V=V\in\Ob\SetV$. Let $f\in\Arr\VecV$, then $f$ is a linear map in $\Definible$,
    so it is a map and then $f\in\Arr\SetV$.
  \end{enumerate}
  Finally, we prove the functors to be monoidal by proving the existence of the
  maps
  \[
    \Span A\otimes\Span B\xlra{m_{AB}}\Span(A\boxtimes B)
    \qquad\textrm{and}\qquad
    I\xlra{m_1}\Span 1
  \]
  in $\Arr{\VecV}$, and
  \[
    \Forg V\boxtimes\Forg W\xlra{n_{VW}} \Forg{(V\otimes W)}
    \qquad\textrm{and}\qquad
    1\xlra{n_I}\Forg I
  \]
  in $\Arr{\SetV}$, trivially satisfying some well-known axioms.
  \begin{itemize}
  \item Since $\Span A\otimes\Span B=\Span{(\Span A\boxtimes\Span
      B)}=\Span{(A\boxtimes B)}$, we take $m_{AB}=\Id$.

  \item Since $I=\Span 1$, we take $m_1=\Id$.

  \item Since $\Forg V\boxtimes \Forg W\subseteq \Forg{\Span{(\Forg
        V\boxtimes\Forg W)}}=\Forg{(V\otimes W)}$, we take $n_{VW}$ as the inclusion map.

  \item Since $1\subseteq\Forg I$, we take $n_I$ as the inclusion map.
    \qedhere
  \end{itemize}
\end{proof}

\noindent
We can now establish the adjunction between the two categories, which will give
us the framework to interpret the calculus.

\begin{prop}\label{prop:adjunction}
  The following construction is a monoidal adjunction:
  \begin{center}
    \begin{tikzcd}[column sep=3mm,labels=description]
      (\SetV,\boxtimes,1)\ar[rr,bend left,"{(\Span{},m)}"] &\bot &
      (\VecV,\otimes,\I)\ar[ll,bend left,"{(\Forg{},n)}"]
    \end{tikzcd}
  \end{center}
  where $m$ and $n$ are mediating arrows of the monoidal structure.
\end{prop}
\begin{proof}
  We need to prove that $\HomV{\Span A,V}\simeq\HomS{A,\Forg{V}}$.
  \begin{itemize}
  \item Let $\HomV{\Span A,V}\xlra{\varphi_{A,V}}\HomS{A,\Forg{V}}$ given by
    $h\mapsto h\upharpoonright A$.
  \item Let $\HomS{A,\Forg{V}}\xlra{\psi_{A,V}}\HomV{\Span A,V}$ given by
    $s\mapsto\Span s$.
  \end{itemize}
  Notice that $\varphi_{A,V}\circ\psi_{A,V}=id$ and
  $\psi_{A,V}\circ\varphi_{A,V}=id$. Therefore, we rename
  $\varphi^{-1}_{A,V}=\psi_{A,V}$.
  We must prove that if $B\xlra f A$, $V\xlra g W$ the following diagrams commute:
  \begin{center}
    \begin{tikzcd}[labels=description,row sep=1.2cm,column sep=1.1cm]
      \HomV{\Span A,V}\ar[r,"\varphi_{A,V}"]\ar[d,swap,"\HomHsfV"] & \HomS{A,\Forg{V}}\ar[d,"\HomSfiV"]
      &[-1.2cm] \HomV{\Span A,V}\ar[r,"\varphi_{A,V}"]\ar[d,swap,"\HomHsAg"] & \HomS{A,\Forg{V}}\ar[d,"\HomSAig"]
      \\
      \HomV{\Span B,V}\ar[r,swap,"\varphi_{B,V}"] &\HomS{B,\Forg{V}}
      &\HomV{\Span A,W}\ar[r,swap,"\varphi_{A,W}"] &\HomS{A,iW}
      \\[-5mm]
      \HomS{A,\Forg{V}}\ar[d,swap,"\HomSfiV"]\ar[r,"\varphi^{-1}_{A,V}"] & \HomV{\Span A,V}\ar[d,"\HomHsfV"]
      & \HomS{A,\Forg{V}}\ar[d,swap,"\HomSAig"]\ar[r,"\varphi^{-1}_{A,V}"] & \HomV{\Span A,V}\ar[d,"\HomHsAg"]
      \\
      \HomS{B,\Forg{V}}\ar[r,swap,"\varphi^{-1}_{B,V}"] &  \HomV{\Span B,V}
      &\HomS{A,iW}\ar[r,swap,"\varphi^{-1}_{A,W}"] &  \HomV{\Span A,W}
    \end{tikzcd}
  \end{center}
  \begin{itemize}
  \item In the first diagram we need to prove that $(h\upharpoonright A)\circ
    f=(h\circ \Span f)\upharpoonright B$.

    We have, for any $b\in B$,
    \begin{align*}
      ((h\upharpoonright A)\circ f)(b)&=h(f(b))\\
                                      &=h(\Span f(b))\\
                                      &=((h\circ \Span f)\upharpoonright B)(b)
    \end{align*}
  \item In the second diagram we need to prove that $\Forg g\circ (h\upharpoonright A)=(g\circ h)\upharpoonright A$.

    We have, for any $a\in A$,
    \begin{align*}
      (\Forg g\circ (h\upharpoonright A))(a)&=g(h(a))\\
                                            &=((g\circ h)\upharpoonright A)(a)
    \end{align*}
  \item The third diagram follows by considering $\Span (s\circ f)=\Span s\circ
    \Span f$.
  \item The last one follows by $\Span (\Forg{g}\circ s)=\Span(\Forg{g})\circ \Span (s)=g\circ \Span (s)$.
  \end{itemize}
  Finally, the monoidality axioms of the adjunction are trivially satisfied.
\end{proof}

Before giving an interpretation of types, we need to define two particular objects $A+B$ and $\home AB$ in
$\SetV$ to interpret the types $A+B$ and $A\rightarrow B$ respectively:
\begin{defi}
  Let $A,B\in\Ob{\SetV}$, we define the following objects.
  \begin{align*}
   A+ B &=\{\Inl{\vec v}\mid\vec v\in A\}\cup\{\Inr{\vec w}\mid\vec w\in B\}\subseteq\ValD\\
   \home AB & =\{\hat f\mid f:A\lra B\in\Arr{\SetV}\}\subseteq\ValD
  \end{align*}
  where $\hat f$ is the term in \OC representing the map $f$.
\end{defi}

In particular, we need to show that $A+B$ is actually a coproduct, as stated by the following lemma.
\begin{lem}%
  \label{lem:coproduct}
  $A+B$ is a coproduct. That is, given $A\xlra f C$ and $B\xlra g C$, there is a unique map $\home fg$ such that the following diagram commutes.
  \begin{center}
    \begin{tikzcd}[labels=description,column sep=2cm,row sep=1.1cm]
      A\ar[r,"i_1"]\ar[rd,"f",swap,sloped] & (A + B)\ar[d,"\home fg",swap]
      &B\ar[l,"i_2",swap]\ar[ld,"g",sloped]\\
      &C&
    \end{tikzcd}
  \end{center}
\end{lem}
\begin{proof}
  Since $A+B \in\Ob\SetV\subset\Ob\Set$ and $f,g\in\Arr\SetV\subset\Arr\Set$, we can take $\home fg\in\Arr\Set$ defined by
  \begin{align*}
    \home fg:A+B&\rightarrow C\\
    x&\mapsto\left\{
      \begin{array}{ll}
        f(\vec a) & \textrm{if }x=\Inl{\vec a}\\
	f(\vec b) & \textrm{ if}x=\Inr{\vec b}
    \end{array}\right.
  \end{align*}
  and prove that $\home fg\in\Arr\SetV$. All we need to prove is that $\hat{\home fg}\in\Definible$. Take $\hat{\home fg}=\lambda x.\Match x {\vec a}{\hat f(\vec a)}{\vec b}{\hat g(\vec b)}$.
\end{proof}

The following lemma allows us to use the home $\home AB$ in the expected way.
\begin{lem}
  There is an adjunction $\_\boxtimes B\dashv\home B{\_}$.
  That is
  \[
    \HomS{A\boxtimes B,C}\simeq\HomS{A,\home BC}
  \]
\end{lem}
\begin{proof}\leavevmode
  \begin{itemize}
  \item Let $A\boxtimes B\xlra{f} C\in\Arr{\SetV}$. Then take
    $A\xlra{\mathsf{curry}f} \home BC$ defined by $\vec a\mapsto\lambda x.\hat f(\vec a,x)$, which can be represented in \OC as $\lambda y.\lambda x.\hat f(y,x)$.

  \item Let $A\xlra{g}\home BC\in\Arr{\SetV}$. Then take $A\boxtimes B\xlra{\mathsf{uncurry}g} C$ defined by $(\vec a,\vec b)\mapsto (g\vec a)\vec b$, with $\hat{\s{uncurry} g}$ being $\lambda x.\LetP{y}{z}{x}{\hat g y z}$.
  \end{itemize}
  Notice that $\hat{\mathsf{uncurry}(\mathsf{curry}f)(a,b)}\to^*\hat f(a,b)$ and $\hat{\mathsf{curry}(\mathsf{uncurry}g)a}\to^* \hat ga$.
  Finally, naturality follows from the fact that $\SetV$ is a subcategory of
  $\mathsf{Set}$. \qedhere
\end{proof}

\begin{defi}
  Types are interpreted in the category $\SetV$, as follows:
  \begin{align*}
    \sem{\mathbb U} & =1\\
    \sem{\sharp A} &= \Forg{\Span{\sem{A}}}\\
    \sem{A+B} &= \sem A+\sem B\\
    \sem{A\times B} &=\sem A\boxtimes\sem B\\
    \sem{A\rightarrow B} &=\home{\sem A}{\sem B}
  \end{align*}
\end{defi}

To avoid cumbersome notation, we write $A$ for $\sem{A}$, when there is no
ambiguity.

Before giving the interpretation of typing derivation trees in the model, we
need to define certain maps that will serve to implement some of the
constructions in the language.

To interpret the \emph{match} construction we define the following map.
\begin{defi}\label{def:if}
  Let $A\boxtimes\Delta\xlra{f}C\in\Arr{\SetV}$ and
  $B\boxtimes\Delta\xlra{g}C\in\Arr{\SetV}$. Then, we define the arrow
  $(A+B)\boxtimes\Delta\xlra{\home fg_1}C$ by
  $(\Inl a,d)\mapsto f(a,d)$ and $(\Inr b,d)\mapsto g(b,d)$.
\end{defi}
%\begin{lem}
%  The map $(A+B)\boxtimes\Delta\xlra{\home fg_1}C$ is in $\Arr{\SetV}$, and it makes the
%  following diagram commutative:
%  \begin{center}
%    \begin{tikzcd}[labels=description,column sep=2cm,row sep=1.1cm]
%      A\boxtimes\Delta\ar[r,"i_1\boxtimes\Id"]\ar[rd,"f",swap,sloped] &
%      (A + B)\boxtimes\Delta\ar[d,"\home fg_1",swap]
%      &B\boxtimes\Delta\ar[l,"i_2\boxtimes \Id",swap]\ar[ld,"g",sloped]\\
%      &C&
%    \end{tikzcd}
%  \end{center}
%\end{lem}
%\begin{proof}
%  We show $\home fg_1\in\Definible$ by taking
%  \(
%    \hat{\home fg_1} =
%    \lambda x.\LetP{y}{d}{x}{
%      \Match{y}{a}{\hat f(a,d)}{b}{\hat g(b,d)}
%      }
%  \).
%  The commutation of the diagram is trivial by evaluation.
%\end{proof}

To interpret the sequence construction, and the rules $\mathsf{Weak}$ and $\mathsf{Contr}$ we need the following maps.
\begin{lem}
  Let $A^\flat$. Then,
  \begin{enumerate}
  \item
    $A\boxtimes B\xlra{\pi_B}B\in\Arr{\SetV}$
    and
    $B\boxtimes A\xlra{\pi_B'}B\in\Arr{\SetV}$.
  \item $A\xlra{\delta}A\boxtimes A\in\Arr{\SetV}$.
  \end{enumerate}
\end{lem}
\needspace{2\baselineskip}
\begin{proof}\leavevmode
  \begin{enumerate}
  \item Take $\hat{\pi_B} := \lambda x.\LetP{y}{z}{x}{z}$, which can be typed as
    follows
    \[
      \infer[{}^{\mathsf{Lam}}]{\vdash\lambda x.\LetP yzxz:A\times B\rightarrow B}
      {
        \infer[{}^{\mathsf{PureLet}}]{x:A\times B\vdash\LetP yzxz:B}
        {
          \infer[{}^{\mathsf{Ax}}]{x:A\times B\vdash x:A\times B}{}
          &
          \infer[{}^{\mathsf{Weak}}]{y:A,z:B\vdash z:B}
          {
            \infer[{}^{\mathsf{Ax}}]{z:B\vdash z:B}{}
            &
            A^\flat
          }
        }
      }
    \]
    $\hat{\pi_B'}$ is analogous.
  \item Take $\hat{\delta} :=\lambda x.(x,x)$, which can be typed as follows
    \[
      \infer[{}^{\mathsf{Lam}}]{\vdash\lambda x.(x,x):A\rightarrow A\times A}
      {
        \infer[{}^{\mathsf{Contr}}]{x:A\vdash (x,x):A\times A}
        {
          \infer[{}^{\mathsf{Pair}}]{x:A,y:A\vdash (x,y):A\times A}
          {
            \infer[{}^{\mathsf{Ax}}]{x:A\vdash x:A}{}
            &
            \infer[{}^{\mathsf{Ax}}]{y:A\vdash y:A}{}
          }
          &
          A^\flat
        }
      }
    \]
  \end{enumerate}
  \vspace{-\baselineskip}\qedhere
\end{proof}

\noindent
We give the interpretation of a type derivation tree in our model. If
$\Gamma\vdash t:A$ with a derivation $\pi$, we write generically $\sem\pi$ as
$\Gamma\xlra{t_A} A$. When $A$ is clear from the context, we may write just $t$
for $t_A$.
\begin{defi}
  If $\pi$ is a type derivation tree, we define $\sem\pi$ inductively as follows,
  \begin{align*}
    &\sem{\vcenter{\infer[^{\mathsf{Ax}}]{x:A\vdash x:A}{}}} = A\xlra{\Id} A\\
    %%%
    &\sem{\vcenter{\infer[^{\mathsf{Lam}}] {\Gamma\vdash\lambda x.\vec t:A\rightarrow B} {\Gamma,x:A\vdash\vec t:B}}} =\Gamma\xlra{\eta^A}[A,\Gamma\boxtimes A]\xlra{[A,\vec t]}[A,B]\\
    %%%
    &\sem{\vcenter{\infer[^{\mathsf{App}}]{\Gamma,\Delta\vdash t\vec s:B} {\Gamma\vdash t:A\rightarrow B & \Delta\vdash\vec s:A }}}
    =\Gamma\boxtimes\Delta\xlra{t\boxtimes\vec s}
     \home AB\boxtimes A
     \xlra{\varepsilon'} B\quad\textrm{where }\varepsilon'=\varepsilon\circ\s{swap} \\
    %%%
    &\sem{\vcenter{\infer[^{\mathsf{Void}}]{\vdash\Void:\mathbb U}{}}} = 1\xlra{\Id}1\\
    %%%
    &\sem{\vcenter{\infer[^{\mathsf{PureSeq}}]{\Gamma,\Delta\vdash\vec t;\vec s:A}{\Gamma\vdash\vec t:\mathbb U & \Delta\vdash\vec s:A}}}
    = \Gamma\boxtimes\Delta\xlra{\vec t\boxtimes\vec s} 1\boxtimes A\xlra{\pi_A} A\\
    %%%
    &\sem{\vcenter{\infer[^{\mathsf{UnitarySeq}}]{\Gamma,\Delta\vdash\vec t;\vec s:\sharp A}{\Gamma\vdash\vec t:\sharp \mathbb U & \Delta\vdash\vec s:\sharp A}}}
    = \Gamma\boxtimes\Delta\xlra{\vec t\boxtimes\vec s} \Forg{\I}\boxtimes \Forg{\Span A}\xlra n \Forg{(\I\otimes\Span A)}\xlra{\Forg{\lambda^{-1}_\otimes}}\Forg{\Span A}\\
    %%%
    &\sem{\vcenter{\infer[^{\mathsf{Pair}}]{\Gamma,\Delta\vdash(\vec v,\vec w):A\times B}{\Gamma\vdash\vec v:A & \Delta\vdash\vec w:B}}} = \Gamma\boxtimes\Delta\xlra{\vec v\boxtimes\vec w} A\boxtimes B\\
    %%%
    &\sem{\vcenter{\infer[^{\mathsf{PureLet}}]{\Gamma,\Delta\vdash\mathsf{let}\ (x,y)=\vec t\ \mathsf{in}\ \vec s:C}{\Gamma\vdash\vec t:A\times B & \Delta,x:A,y:B\vdash\vec s:C}}} \\
    &=
      \Gamma\boxtimes\Delta
      \xlra{\vec t\boxtimes\eta^{A\boxtimes B}} (A\boxtimes B)\boxtimes\home{A\boxtimes B}{\Delta\boxtimes A\boxtimes B}
      \\
    &\xlra{\Id\boxtimes\home{A\boxtimes B}{\vec s}} (A\boxtimes B)\boxtimes\home{A\boxtimes B}C\xlra{\varepsilon}C
    \\
    %%%
    &\sem{\vcenter{\infer[^{\mathsf{UnitaryLet}}]{\Gamma,\Delta\vdash\mathsf{let}\ (x,y)=\vec t\ \mathsf{in}\ \vec s:\sharp C}{\Gamma\vdash\vec t:A\otimes B & \Delta,x:\sharp A,y:\sharp B\vdash\vec s:\sharp C}}} \\
    &=\Gamma\boxtimes\Delta
    \xlra{\vec t\boxtimes\eta^{\Forg{\Span A}\boxtimes\Forg{\Span B}}}\Forg{\Span{(A\boxtimes B)}}\boxtimes\home{\Forg{\Span A}\boxtimes\Forg{\Span B}}{\Delta\boxtimes\Forg{\Span A}\boxtimes\Forg{\Span B}}\\
    &\xlra{\Id\boxtimes\home{\Forg{\Span A}\boxtimes\Forg{\Span B}}{\vec s}}\Forg{\Span{(A\boxtimes B)}}\boxtimes\home{\Forg{\Span{A}}\boxtimes \Forg{\Span B}}{\Forg{\Span{C}}}\\
    & \xlra{\Id\boxtimes\eta} \Forg{\Span{(A\boxtimes B)}}\boxtimes\Forg{\Span{\home{\Forg{\Span A}\boxtimes \Forg{\Span B}}{\Forg{\Span C}}}}\\
    &\xlra{n} \Forg{(\Span{(A\boxtimes B)}\otimes\Span{\home{\Forg{\Span{A}}\boxtimes \Forg{\Span B}}{\Forg{\Span{C}}}})}\\
    &\xlra{Um} \Forg{\Span{((A\boxtimes B)\boxtimes\home{\Forg{\Span{A}}\boxtimes \Forg{\Span B}}{\Forg{\Span{C}}})}}\\
    &\xlra{\Forg{\Span{((\eta\boxtimes\eta)\boxtimes\Id)}}} \Forg{\Span{((\Forg{\Span{A}}\boxtimes \Forg{\Span{B}})\boxtimes\home{\Forg{\Span{A}}\boxtimes \Forg{\Span B}}{\Forg{\Span{C}}})}}
    \xlra{\Forg{\Span{\varepsilon}}} \Forg{\Span{\Forg{\Span C}}}\xlra\mu \Forg{\Span C}\\
    %%%
    &\sem{\vcenter{\infer[^{\mathsf{InL}}]{\Gamma\vdash\mathsf{inl}(\vec v):A+B}{\Gamma\vdash\vec v:A}}} = \Gamma\xlra{\vec v} A\xlra{i_1} A+B\\
    %%%
    &\sem{\vcenter{\infer[^{\mathsf{InR}}]{\Gamma\vdash\mathsf{inr}(\vec v):A+B}{\Gamma\vdash\vec v:B}}} = \Gamma\xlra{\vec v} B\xlra{i_2} A+B\\
    %%%
    &\sem{\vcenter{\infer[^{\mathsf{PureMatch}}]{\Gamma,\Delta\vdash\mathsf{match}\ \vec t\ \{\mathsf{inl}(x_1)\mapsto\vec v_1\mid\mathsf{inr}(x_2)\mapsto\vec v_2\}:C}{\Gamma\vdash\vec t:A+B & \Delta\vdash(x_1:A\vdash\vec v_1\perp x_2:B\vdash\vec v_2):C}}}\\
    &=\Gamma\boxtimes\Delta\xlra{\vec t\boxtimes\Id} (A+B)\boxtimes\Delta\xlra{\home{\vec v_1}{\vec v_2}_1} C
    \\
    %%%
    &\sem{\vcenter{\infer[^{\mathsf{UnitaryMatch}}]{\Gamma,\Delta\vdash\mathsf{match}\ \vec t\ \{\mathsf{inl}(x_1)\mapsto\vec v_1\mid\mathsf{inr}(x_2)\mapsto\vec v_2\}:\sharp C}
      {
        \Gamma\vdash\vec t:A\oplus B
	&
        \Delta\vdash(x_1:\sharp A\vdash\vec v_1\perp x_2:\sharp B\vdash\vec v_2):\sharp C
    }}}\\
    &=\Gamma\boxtimes\Delta\xlra{\vec t\boxtimes\Id} \Forg{\Span{(A + B)}}\boxtimes\Delta
    \xlra{\Forg{\Span{(\eta+\eta)}}\boxtimes \eta} \Forg{\Span{(\Forg{\Span{A}}+\Forg{\Span{B}})}}\boxtimes \Forg{\Span{\Delta}}\\
    &\xlra n \Forg{(\Span{(\Forg{\Span{A}}+\Forg{\Span{B}})}\otimes\Span{\Delta})}
    \xlra{\Forg{m}} \Forg{\Span{((\Forg{\Span{A}}+\Forg{\Span{B}})\boxtimes\Delta)}}
      \xlra{\Forg{\Span{\home{\vec v_1}{\vec v_2}_1}}} \Forg{\Span{C}}
    \\
    %%%
    &\sem{\vcenter{\infer[^{\mathsf{Sup}}]{\vdash\sum_{j=1}^m\alpha_j\cdot\vec v_j:\sharp A}
      {
        \text{\scriptsize $(k\neq h)$}
	&
        \vdash(\vec v_k\perp\vec v_h):A
	&
        \sum_{j=1}^m|\alpha_j|^2=1
	&
	\text{\scriptsize $m\geq 1$}
	&
        A\neq B\rightarrow C
    }}}\\
    &=1\xlra{\eta}\Forg{\Span 1}\xlra{\Forg{\sum_j\alpha_j\cdot\Span{\vec v_j}}} \Forg{\Span{A}}
    \\
    %%%
    &\sem{\vcenter{\infer{A\leq A}{}}} = A\xlra{\Id}A \\
    &\sem{\vcenter{\infer{A\leq C}{A\leq B & B\leq C}}} = A\xlra{\sem{A\leq B}}B\xlra{\sem{B\leq C}}C \\
    &\sem{\vcenter{\infer{A\leq\sharp A}{}}} = A\xlra{\eta}\Forg{\Span A}\\
    &\sem{\vcenter{\infer{\sharp \sharp A\leq\sharp A}{}}} =
      \Forg{\Span{\Forg{\Span A}}}\xlra{\mu}\Forg{\Span A} \quad\textrm{where}\ \mu=\Forg{\varepsilon_{\Span{}}}\\
      &\sem{\vcenter{\infer{A'\rightarrow B\leq A\rightarrow B'}{A\leq A' & B\leq B'}}} = \home{A'}B\xlra{\home{\sem{A\leq A'}}{\sem{B\leq B'}}}\home A{B'}\\
    &\sem{\vcenter{\infer{A\times B\leq A'\times B'}{A\leq A' & B\leq B'}}} =
                                                                A\boxtimes B\xlra{\sem{A\leq A'}\boxtimes\sem{B\leq B'}}A'\boxtimes B'\\
    &\sem{\vcenter{\infer{A+B\leq A'+B'}{A\leq A' & B\leq B'}}} = A+B\xlra{\sem{A\leq A'}+\sem{B\leq B'}}A'+B'\\
    &\sem{\vcenter{\infer[^{\leq}]{\Gamma\vdash\vec t:B}{\Gamma\vdash\vec t:A & A\leq B}}} = \Gamma\xlra{\vec t}A\xlra{\sem{A\leq B}}B\\
    &\sem{\vcenter{\infer[^{\equiv}]{\Gamma\vdash\vec r:A}{\Gamma\vdash\vec t:A & \vec t\equiv\vec r}}} = \Gamma\xlra{\vec t} A\\
    &\sem{\vcenter{\infer[^{\mathsf{Weak}}]{\Gamma,x:A\vdash\vec t:B}{\Gamma\vdash\vec t:B & A^\flat}}} = \Gamma\boxtimes A\xlra{\pi_\Gamma'} \Gamma\xlra{\vec t} B\\
    &\sem{\vcenter{\infer[^{\mathsf{Contr}}]{\Gamma,x:A\vdash\vec t[y:=x]:B}{\Gamma,x:A,y:A\vdash\vec t:B & A^\flat}}}
    = \Gamma\boxtimes A\xlra{\Id\boxtimes\delta} \Gamma\boxtimes A\boxtimes A\xlra{\vec t} B
  \end{align*}
\end{defi}

Lemma~\ref{lem:eqDer} allows us to write the semantics of a sequent,
independently of its derivation. Hence, due to this independence, we can write
$\sem{\Gamma\vdash t:A}$, without ambiguity.

\begin{lem}[Independence of derivation]\label{lem:eqDer}
  If $\Gamma\vdash t:A$ can be derived with two different derivations $\pi$ and
  $\pi'$, then $\sem\pi=\sem{\pi'}$.
\end{lem}
\begin{proof}
  We first give a rewrite system on derivation trees such that if one rule can
  be applied after or before another rule, we choose a direction to rewrite the
  tree to one of these forms.
  Then, we prove that every rule preserves the semantics of the tree. This
  rewrite system is clearly confluent and normalizing, hence for each tree $\pi$
  we can take the semantics of its normal form, and so every sequent will have
  one way to calculate its semantics: as the semantics of the normal tree.

  The introduction rules ($\s{Ax}$, $\s{Lam}$, $\s{Void}$, $\s{Pair}$,
  $\s{InL}$, $\s{InR}$, and $\s{Sup}$) are syntax-directed. So, whenever we have
  an introduction term, we know precisely what is the last rule applied in its
  derivation tree.

  For the structural rules ($\leq$, $\equiv$, $\s{Weak}$, and $\s{Contr}$) if
  they can be applied as last rule, or can be applied before, then we can always
  choose to rewrite the tree to apply them at the end, and also choose and order
  between them.

  Finally, the elimination rule $\s{App}$ is the only one which is
  syntax-directed, all the others are not, since for each
  elimination term there exists a rule preceded with $\s{Pure}$ and another with
  $\s{Unitary}$. However, if a term can be typed with one, cannot be typed with
  the other, except for $\s{Seq}$, which, in combination with
  $\leq$, can be interchanged, and then we have to choose a direction to rewrite
  one to the other:
  \[
    \vcenter{\infer[^{\s{UnitarySeq}}]{\Gamma,\Delta\vdash\vec t;\vec s:\sharp A}
      {
        \infer[^\leq]{\Gamma\vdash\vec t:\sharp\U}{\Gamma\vdash\vec t:\U}
        &
        \infer[^\leq]{\Delta\vdash\vec s:\sharp A}{\Delta\vdash\vec s:A}
      }
    }
    \qquad\longrightarrow\qquad
    \vcenter{
      \infer[^\leq]{\Gamma,\Delta\vdash\vec t;\vec s:\sharp A}
      {
        \infer[^{\s{PureSeq}}]{\Gamma,\Delta\vdash\vec t;\vec s:A}
        {
          \Gamma\vdash\vec t:\U
          &
          \Delta\vdash\vec s:A
        }
      }
    }
  \]
  The confluence of this rewrite system is easily inferred from the fact that
  there are not critical pairs. The normalization follows from the fact that the
  trees are finite and all the rewrite rules push the structural rules to the
  root of the trees.

  It only remains to check that each rule preserves the semantics.
  \begin{itemize}
    \item The structural rules follow trivially by naturality.
    \item $\s{UnitarySeq}\lra\s{PureSeq}$
      The diagrams for the left-side and the right-side of the rewrite rule,
      commutes as it is shown
      below:
  \end{itemize}
  
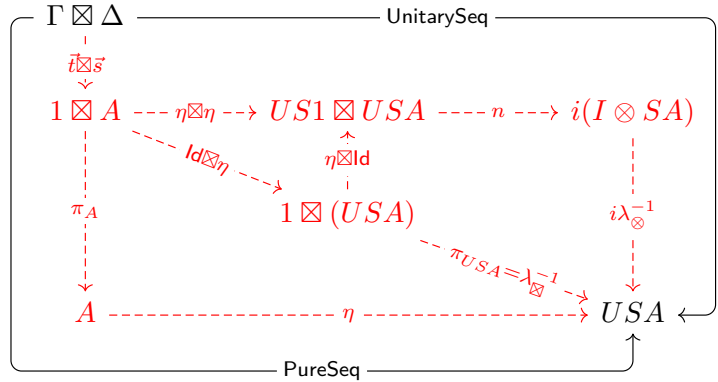
\begin{center}
      \begin{tikzcd}[labels=description,column sep=0.8cm]
	\Gamma\boxtimes\Delta
	\arrow[rrrrddd, "\s{UnitarySeq}" description,
	  rounded corners,
	  to path={[pos=0.25]
	    -- ([xshift=3mm]\tikztostart.east)
	    -| ([xshift=5mm]\tikztotarget.east)\tikztonodes
	  -- (\tikztotarget)}
	]
	\arrow[rrrrddd, "\s{PureSeq}" description,
	  rounded corners,
	  to path={[pos=0.75]
	    -- ([xshift=-3mm]\tikztostart.west)
	    |- ([yshift=-5mm]\tikztotarget.south)\tikztonodes
	  -- (\tikztotarget)}
	]
	\arrow[d, "\vec t\boxtimes\vec s", dashed, red] \\
	 \color{red}1\boxtimes A\arrow[rrd,"\Id\boxtimes\eta",sloped,dashed,red] \arrow[rr, "\eta\boxtimes\eta", dashed, red] \arrow[dd, "\pi_A", dashed, red]  & &
	\color{red}\Forg{\Span{1}}\boxtimes \Forg{\Span{A}} \arrow[rr, "n", dashed, red] & & \color{red}i(I\otimes\Span A) \arrow[dd, "i\lambda^{-1}_\otimes", dashed, red] \\
	 & &\color{red}1\boxtimes (\Forg{\Span{A}})\arrow[u,"\eta\boxtimes\Id",dashed,red]\arrow[rrd,"\pi_{\Forg{\Span{A}}}=\lambda^{-1}_\boxtimes",sloped,dashed,red]  & &   \\
	 \color{red}A \arrow[rrrr, "\eta", dashed, red] & & & & \Forg{\Span{A}} &
      \end{tikzcd}
  \end{center}
  \vspace{-\baselineskip}\qedhere
\end{proof}

\section{Soundness and (partial) completeness}\label{sec:SC}
We prove the soundness of our interpretation with respect to reduction, and the
completeness only on type $\sharp(\U+\U)$, which corresponds to $\C^2$.

\begin{lem}
  [Substitution]\label{lem:substitution}
  If $\Gamma,x:A\vdash\vec t:B$ and $\Delta\vdash v:A$, then the following diagram commutes:
  \begin{center}
    \begin{tikzcd}[labels=description]
      \Gamma\boxtimes\Delta\ar[rr,"\vec t{[}x:=v{]}"]\ar[dr,"\Id\boxtimes v",sloped] & & B\\
      &\Gamma\boxtimes A\ar[ur,"\vec t",sloped] &
    \end{tikzcd}
  \end{center}
  That is,
  $\sem{\Gamma,\Delta\vdash\vec t[x:=v]:B}=\sem{\Gamma,x:A,\Gamma\vdash\vec t:B}\circ(\Id\boxtimes\sem{\Delta\vdash v:A})$.
\end{lem}
\begin{proof}
  By induction on the derivation of $\Gamma,x:A\vdash\vec t:B$.
  The details can be found in~\ref{app:soundness}
  \qedhere
\end{proof}

\begin{thm}
  [Soundness]\label{thm:soundness}
  If $\Gamma\vdash t:A$, and $t\lra r$,
  then
  $\sem{\Gamma\vdash t:A} = \sem{\Gamma\vdash r:A}$.
\end{thm}
\begin{proof}
  By induction on the rewrite relation. The details can be found in~\ref{app:soundness}.
\end{proof}

\begin{lem}[Completeness of values on $\mathbb C^2$]\label{lem:completenessvalues}
  If $\sem{\vdash\vec v:\sharp(\U+\U)}=\sem{\vdash\vec w:\sharp(\U+\U)}$, then $\vec v\equiv\vec w$.
\end{lem}
\begin{proof}
  Since $\vec v$ and $\vec w$ are values of type $\sharp(\U+\U)$, they have the
  following shape:
  $\vec v$ is $\equiv$-equivalent to either
  \begin{enumerate}
    \item $\alpha\cdot\Inl\Void+\beta\cdot\Inr\Void$,
    \item $\Inl\Void$, or
    \item $\Inr\Void$,
  \end{enumerate}
  and
  $\vec w$ is $\equiv$-equivalent to either
  \begin{enumerate}
    \item $\gamma\cdot\Inl\Void+\delta\cdot\Inr\Void$,
    \item $\Inl\Void$, or
    \item $\Inr\Void$.
  \end{enumerate}
  Indeed, we have consider three cases since $1\cdot\Inl\Void+0\cdot\Inr\Void\not\equiv\Inl\Void$ because $\ValD$ is
  a distributive-action space and not a vector space.

  We analyse the different cases:
  \begin{itemize}
  \item If $\vec v$ and $\vec w$ are both in case $1$, i.e.~$\vec v\equiv
    \alpha\cdot\Inl\Void+\beta\cdot\Inr\Void$ and $\vec w\equiv\gamma\cdot\Inl\Void+\delta\cdot\Inr\Void$, we have
    \begin{align*}
      \sem{\vdash\vec v:\sharp(\U+\U)}
      &=1\xlra{\eta}US1\xlra{U(\alpha\cdot Si_1+\beta\cdot Si_2)}\sharp(\U+\U)\\
      \sem{\vdash\vec w:\sharp(\U+\U)}
      &=1\xlra{\eta}US1\xlra{U(\gamma\cdot Si_1+\delta\cdot Si_2)}\sharp(\U+\U)
    \end{align*}
    So, the maps $\Void\mapsto\alpha\cdot\Inl\Void+\beta\cdot\Inr\Void$ and
    $\Void\mapsto\gamma\cdot\Inl\Void+\delta\cdot\Inr\Void$ are the same, and so,
    since $\Inl\Void\perp\Inr\Void$, we have
    $\alpha=\gamma$ and $\beta=\delta$, thus, $\vec v\equiv\vec w$.
  \item If $\vec v$ and $\vec w$ are both in case $2$ or both in case $3$, then
    $\vec v\equiv\vec w$.
  \item It is easy to see that $\vec v$ and $\vec w$ cannot be in different cases, since
    $\sem{\vdash\vec v:\sharp(\U+\U)}$ must be equal to $\sem{\vdash\vec w:\sharp(\U+\U)}$ and this is not the case when they are in different
    cases. For example, let $\vec v=\Inl\Void$ and $\vec w=1\cdot\Inl\Void+0\cdot\Inr\Void$. In this case we have
    \begin{align*}
      \sem{\vdash\vec v:\sharp(\U+\U)}
      &=1\xlra{i_1}1+1\xlra{\eta}\sharp(\U+\U)\\
      \sem{\vdash\vec w:\sharp(\U+\U)}
      &=1\xlra{\eta}US1\xlra{U(1\cdot Si_1+0\cdot Si_2)}\sharp(\U+\U)
    \end{align*}
    the first arrow with the mapping $\Void\mapsto\Inl\Void$ while the second
    $\Void\mapsto 1\cdot\Inl\Void+0\cdot\Inr\Void$, which are not equal in
    a distributive-action space
    \qedhere
  \end{itemize}
\end{proof}

\begin{defi}
  [Equivalence on terms]
  We write $\vec t\sim\vec r$ whenever $\vec t\lra^*\vec v$ and $\vec r\lra^*\vec w$, with $\vec v\equiv\vec w$.
\end{defi}

\begin{thm}[Completeness on $\mathbb C^2$]\label{thm:completeness}
  If $\sem{\vdash\vec t:\sharp(\U+\U)}=\sem{\vdash\vec r:\sharp(\U+\U)}$, then $\vec t\sim\vec r$.
\end{thm}
\begin{proof}
  By progress (Theorem~\ref{thm:progress}) and strong normalization
  (Corollary~\ref{cor:SN}), we have that $\vec t\lra^*\vec v$ and $\vec r\lra^*\vec w$. By subject reduction (Theorem~\ref{thm:SR}), we have $\vdash\vec v:\sharp(\U+\U)$ and $\vdash\vec w:\sharp(\U+\U)$. Hence, by soundness
  (Theorem~\ref{thm:soundness}), $\sem{\vdash\vec v:\sharp(\U+\U)}=\sem{\vdash
    t:\sharp(\U+\U)}=\sem{\vdash r:\sharp(\U+\U)}=\sem{\vdash w:\sharp(\U+\U)}$.
  Thus, by the completeness of values on $\mathbb C^2$
  (Lemma~\ref{lem:completenessvalues}), we have $\vec v\equiv\vec w$. So, by
  definition, $\vec t\sim\vec r$.
\end{proof}

\begin{thm}
  [Completeness on qubits]\label{thm:completenessonqubits}
  If $\sem{\vdash\vec t:\mathbb B^{\otimes n}}=\sem{\vdash\vec r:\mathbb
    B^{\otimes n}}$, then $\vec t\sim\vec r$.
\end{thm}
\begin{proof}
  We prove it for $n=2$. The generalization is straightforward. Since $\mathbb
  B^{\otimes 2}=\sharp(\mathbb B\times\mathbb B)=\sharp((\U+\U)\times(\U+\U))$,
  the set of closed values with this type is
  \[
    Q_2 = \{\sum_i\alpha_i\cdot(v_i,w_i)\mid v_i,w_i\in\{\Inl\Void,\Inr\Void\}\}
  \]
  Hence, if $\vec v\in Q_2$, $\vec v$ is $\equiv$-equivalent to one of
  \begin{align*}
    &\alpha_0\cdot v_0+\alpha_1\cdot v_1+\alpha_2\cdot v_2+\alpha_3\cdot v_3\\
    &\alpha_0\cdot v_0+\alpha_1\cdot v_1+\alpha_2\cdot v_2\\
    &\alpha_0\cdot v_0+\alpha_1\cdot v_1\\
    &\alpha_0\cdot v_0
  \end{align*}
  with $v_i\in\{(w_1,w_2)\mid w_1,w_2\in\{\Inl\Void,\Inr\Void\}\}$.

  Following the same reasoning from Lemma~\ref{lem:completenessvalues}, we get
  that if $\sem{\vdash\vec v:\mathbb B^{\otimes 2}}=\sem{\vdash\vec w:\mathbb
    B^{\otimes 2}}$, then $\vec v\equiv\vec w$, and following the same reasoning
  from Theorem~\ref{thm:completeness}, we have that if $\sem{\vdash\vec t:\mathbb
    B^{\otimes 2}}=\sem{\vdash\vec r:\mathbb B^{\otimes 2}}$, then $\vec t\sim\vec r$.
\end{proof}

\section{Conclusion}\label{sec:Conclusion}
In this paper we have introduced \OC, a quantum calculus issued from a
realizability semantics~\cite{DiazcaroGuillermoMiquelValironLICS19}, and
presented a categorical construction to interpret it.
\paragraph{Comparison with Lambda-$\mathcal S$} The main difference between
\OC and Lambda-$\mathcal
S$~\cite{DiazcaroDowekTPNC17,DiazcaroDowekRinaldiBIO19} is the fact that \OC
enforces norm $1$ vectors by defining a distributive-action space on values. This gives
us a related but different model than those for Lambda-$\mathcal
S$~\cite{DiazcaroMalherbeLSFA18,DiazcaroMalherbe20,DiazcaroMalherbeACS20}. In
addition, in \OC it is allowed to type a superposition with a product type,
whenever it is a separable state. Indeed, since $(\vec v,\vec w):=\sum_{i=1}^n\sum_{j=1}^m\alpha_i\beta_j\cdot(v_i,w_j)$. We have
\[
  \infer{\Gamma,\Delta\vdash(\vec v,\vec w):A\times B} {\Gamma\vdash\vec v:A &
    \Delta\vdash\vec w:B}
\]
That is,
\[\vcenter{
  \infer
  {\Gamma,\Delta\vdash\sum_{i=1}^n\sum_{j=1}^m\alpha_i\beta_j\cdot(v_i,w_j):A\times B}
{\Gamma\vdash\sum_{i=1}^n v_i:A & \Delta\vdash\sum_{j=1}^m w_j:B}}
\]
So, the type $A\times B$ is telling us that the term is separable, while a
generic term of the form $\sum_{i=1}^n\sum_{j=1}^m\alpha_i\beta_j\cdot(v_i,w_j)$
would have type $\sharp(A\times B)=A\otimes B$. In Lambda-$\mathcal S$ the only
way to type such a term is with a tensor, and hence we need a casting operator
in order to lose separability information when needed to allow reduction without
losing type preservation.

The main property of \OC is the fact that any isometry can be represented in the
calculus (Theorem~\ref{thm:expressiveness}) and any term of type
$\sharp(\U+\U)\rightarrow\sharp(\U+\U)$ represents an isometry
(Theorem~\ref{thm:isometryLICS}).

\paragraph{Comparison with the full calculus}
In its original presentation~\cite{DiazcaroGuillermoMiquelValironLICS19}, any
arbitrary type $A$ is defined by its semantics $\semR{\cdot}$ as a set of
values, and the notation $\vdash t:A$ means that $t$ reduces to a value in
$\semR A$ (notation $t\Vdash A$). For example, let
\(
  \ket +=\tfrac 1{\sqrt 2}\cdot\mathsf{inl}(\Void)+\tfrac 1{\sqrt 2}\cdot\mathsf{inr}(\Void)
\).
We can consider $\semR{A}=\{\ket +\}\subseteq\semR{\sharp(\mathbb U+\mathbb U)}=\semR{\sharp\mathbb B}$
(even if this is not a type we can construct with the given syntax of types).

The realizability semantics is so strong that it even allows defining a set
$\semR{A\Rightarrow B}$ of linear combinations of abstractions. However, not
every linear combination of values in $\semR{A\rightarrow B}$ is valid in the
semantics. Indeed, if $\lambda x.\vec t$ and $\lambda x.\vec s$ are both in
$\semR{A\rightarrow B}$, the linear combination $\alpha\cdot\lambda x.\vec t+\beta\cdot\lambda x.\vec s$ with $|\alpha|^2+|\beta^2|$ will be in
$\semR{A\Rightarrow B}$, if and only if $\alpha\cdot\vec t[v/x]+\beta\cdot\vec s[v/x]$ have norm
$1$ for any $v\in\semR A$. Hence, $\semR{A\Rightarrow B}$ is contained, but not
equal to $\semR{\sharp(A\rightarrow B)}$.

The set $\semR{A\Rightarrow B}$ can be easily constructed by the realizability
semantics (since the typing $\vdash\vec t:C$ is done by reducing the term $\vec t$ and checking that the resulted value is in $\semR C$), but not with static
methods. Therefore, we decided to exclude the type $A\Rightarrow B$ in \OC. % chktex 13
Remark that even without superposition of abstractions, we do not lose
expressivity, since $\alpha\cdot\lambda x.\vec t+\beta\cdot\lambda x.\vec s$
behaves in the same way that $\lambda x.(\alpha\cdot\vec t+\beta\cdot\vec s)$,
which can be typed in \OC. % chktex 13

%\section*{Acknowledgement}
%  \noindent The authors wish to acknowledge fruitful discussions with
%  A and B.

\bibliographystyle{alphaurl}
\bibliography{biblio}

\appendix
\section{Soundness}\label{app:soundness}
\noindent{\bf Lemma~\ref{lem:substitution}} (Substitution){\textbf{.}}{\itshape If
  $\Gamma,x:A\vdash\vec t:B$ and $\Delta\vdash v:A$, then the following diagram
  commutes:
  \begin{center}
    \begin{tikzcd}
      \Gamma\boxtimes\Delta\ar[rr,"\vec t{[}x:=v{]}" description]\ar[dr,"\Id\boxtimes
      v" description,sloped,swap] & & B\\
      &\Gamma\boxtimes A\ar[ur,"\vec t" description,sloped,swap] &
    \end{tikzcd}
  \end{center}
  That is, $\sem{\Gamma,\Delta\vdash\vec t[x:=v]:B}=\sem{\Gamma,x:A,\Gamma\vdash\vec t:B}\circ(\Id\boxtimes\sem{\Delta\vdash v:A})$.\/}
\begin{proof}
  By induction on the derivation of $\Gamma,x:A\vdash\vec t:B$.
  \begin{itemize}
  \item ${\vcenter{\infer[^{\mathsf{Ax}}]{x:A\vdash x:A}{}}}$
    \begin{center}
      \begin{tikzcd}
        1\boxtimes\Delta\cong\Delta\ar[rr,"v" description]\ar[dr,"\Id\boxtimes v" description,sloped,swap] & & A\\
        & 1\boxtimes A\cong A\ar[ur,"\Id" description,sloped,swap] &
      \end{tikzcd}
    \end{center}

  \item ${\vcenter{\infer[^\leq]{\Gamma,x:A\vdash\vec t:C}{\Gamma,x:A\vdash\vec t:B & B\leq C}}}$
    \begin{center}
      \begin{tikzcd}[column sep=3cm,row sep=1cm,execute at end picture={
          \path (\tikzcdmatrixname-1-1) -- (\tikzcdmatrixname-1-3)
          coordinate[pos=0.5](aux) (aux) -- (\tikzcdmatrixname-2-2)
          node[midway,blue]{\small (Def)}; \path (\tikzcdmatrixname-1-1) --
          (\tikzcdmatrixname-3-2) coordinate[pos=0.5](aux) (aux) --
          (\tikzcdmatrixname-2-2) node[midway,blue]{\small (IH)}; \path
          (\tikzcdmatrixname-3-2) -- (\tikzcdmatrixname-1-3)
          coordinate[pos=0.5](aux) (aux) -- (\tikzcdmatrixname-2-2)
          node[midway,blue]{\small (Def)}; }]
        \Gamma\boxtimes\Delta \arrow[rdd,sloped,"\Id\boxtimes v" description,swap] \arrow[rr,"{\vec{t}[x:=v]^C}" description] \arrow[rd, "{\vec t[x:=v]^B}" description,sloped,dashed,red] & & C \\
        &  \color{red}B \arrow[ru, "{\sem{B\leq C}}" description,sloped,dashed,red] & \\
        & \Gamma\boxtimes A \arrow[ruu, "\vec{t}^C" description,swap,sloped]
        \arrow[u,"\vec{t}^B" description,dashed,red] &
      \end{tikzcd}
    \end{center}

  \item $\vcenter{\infer[^{\equiv}]{\Gamma\vdash\vec r:A}{\Gamma\vdash\vec t:A &
        \vec t\equiv\vec r}}$

    \begin{center}
      \begin{tikzcd} [labels=description,execute at end picture={
          \path (\tikzcdmatrixname-1-1) -- (\tikzcdmatrixname-1-3) coordinate[pos=0.5](aux)
          (\tikzcdmatrixname-2-2) -- (aux) node[midway,blue]{\small (IH)}
          (\tikzcdmatrixname-2-2) -- (aux) node[pos=1.4,blue]{\small (Def)}
          (\tikzcdmatrixname-2-2) -- (\tikzcdmatrixname-1-3) node[pos=0.52,yshift=-3mm,sloped,blue]{\small (Def)}
          ;
        }]
        \Gamma\boxtimes\Delta\ar[rr,"\vec t{[}x:=v{]}",dashed,red]\ar[rr,"\vec r{[x:=v]}",bend left]\ar[dr,"\Id\boxtimes v",sloped] & & A\\
        &\Gamma\boxtimes A\ar[ur,"\vec t",sloped,dashed,red,bend left=10] \ar[ur,"\vec r",sloped,out=0,in=-90] &
      \end{tikzcd}
    \end{center}

  \item \needspace{4em}
    ${\vcenter{\infer[^{\mathsf{Lam}}] {\Gamma,x:A\vdash\lambda y.\vec t:B\rightarrow C} {\Gamma,x:A,y:B\vdash\vec t:C}}}$
    \begin{center}
      \begin{tikzcd}[labels=description,column sep=3cm,row sep=1.1cm,execute at
        end picture={
          \path (\tikzcdmatrixname-1-1) -- (\tikzcdmatrixname-3-2)
          node[midway,blue,sloped]{\small (Naturality of $\eta^B$)};
          \path
          (\tikzcdmatrixname-2-2) -- (\tikzcdmatrixname-1-3)
          coordinate[pos=0.5](aux) (\tikzcdmatrixname-1-2) -- (aux)
          node[midway,blue]{\tiny (IH \& functoriality)}; \path (\tikzcdmatrixname-3-2) --
          (\tikzcdmatrixname-1-3) node[midway,blue,sloped]{\small (Def)};
          \path (\tikzcdmatrixname-1-1) -- (\tikzcdmatrixname-1-3) node[midway,yshift=8mm,blue]{\small (Def)}; }]
        \Gamma\boxtimes\Delta \arrow[rdd, "\Id\boxtimes v"',bend right,sloped] \arrow[rr, "{\lambda y.\vec t[x:=v]}", bend left] \arrow[r, "\eta^B", dashed,red] & {\color{red}[B,\Gamma\boxtimes\Delta\boxtimes B]} \arrow[r, "{[B,\vec t[x:=v]]}", dashed,red] \arrow[d, "{[B,\Id\boxtimes v\boxtimes \Id]}"', dashed,red] & {[B,C]} \\
        & {\color{red}[B,\Gamma\boxtimes A\boxtimes B]} \arrow[ru, "{[B,\vec t]}"', dashed,red,sloped] &         \\
        & \Gamma\boxtimes A \arrow[u, "\eta^B", dashed,red] \arrow[ruu, "\lambda
        y.\vec t"',bend right,sloped] &
      \end{tikzcd}
    \end{center}

  \item $\vcenter { \infer[^{\mathsf{App}}]
      {\Gamma,\Xi,x:A\vdash t\vec s:C}
      {
        \Gamma\vdash t:B\rightarrow C
        &
        \Xi,x:A\vdash\vec s:B
      }
    }$
    \begin{center}
      \begin{tikzcd}[labels=description,column sep=2cm,row sep=1.2cm,
        execute at end picture={
          \path (\tikzcdmatrixname-1-1) -- (\tikzcdmatrixname-1-3) coordinate[pos=0.5](aux1)
          (aux1) -- (\tikzcdmatrixname-2-2) node[midway,blue]{\small (Def)};
          \path (\tikzcdmatrixname-1-1) -- (\tikzcdmatrixname-3-2) node[sloped,midway,blue]{\small (HI)};
          \path (\tikzcdmatrixname-3-2) -- (\tikzcdmatrixname-1-3) node[sloped,midway,blue]{\small (Def)};
           }]
        \Gamma\boxtimes\Xi\boxtimes\Delta \arrow[rdd, "\Id\boxtimes v" , bend right, sloped] \arrow[rd, "{t\boxtimes\vec s[x:=v]}" , sloped, dashed,red] \arrow[rr, "{t\vec s[x:=v]}" , sloped] && C \\
        & {\color{red}[B,C]\boxtimes B} \arrow[ru, "\varepsilon'" , sloped, dashed,red]                                                     &   \\
        & \Gamma\boxtimes\Xi\boxtimes A \arrow[u, "t\boxtimes\vec s" ,
        dashed,red] \arrow[ruu, "t\vec s" , sloped,bend right] &
      \end{tikzcd}
    \end{center}

    The case where $x\in\FV(t)$ is analogous.

  \item ${\vcenter{\infer[^{\mathsf{Void}}]{\vdash\Void:\mathbb U}{}}}$

    This case does not follow the hypothesis of the lemma.

  \item ${ \vcenter{ \infer[^{\mathsf{PureSeq}}] {\Gamma,\Xi,x:A\vdash\vec t;\vec s:B} { \Gamma\vdash\vec t:\mathbb U & \Xi,x:A\vdash\vec s:B}}}$
    \begin{center}
      \begin{tikzcd}[labels=description,row sep=1.2cm,column sep=1.2cm ,execute at end picture={
           \path (\tikzcdmatrixname-1-1) -- (\tikzcdmatrixname-1-3) coordinate[pos=0.5](aux1)
           (aux1) -- (\tikzcdmatrixname-2-2) node[midway,blue]{\small (Def)};
           \path (\tikzcdmatrixname-3-2) -- (\tikzcdmatrixname-1-3) node[pos=0.3,blue]{\small (Def)};
           \path (\tikzcdmatrixname-1-1) -- (\tikzcdmatrixname-3-2) node[midway,blue]{\small (IH)};
         }
        ]
        \Gamma\boxtimes\Xi\boxtimes\Delta \arrow[dr, "{\vec t\boxtimes\vec s[x:=v]}", dashed,red,sloped] \arrow[ddr, "\Id\boxtimes v"',sloped,bend right] \arrow[rr, "{\vec t;\vec s[x:=v]}"]  & & B \\
        & {\color{red}1\boxtimes B} \arrow[ru, "\pi_B",sloped,dashed,red] &   \\
        & \Gamma\boxtimes\Xi\boxtimes A \arrow[ruu, "\vec t;\vec s"',sloped,bend right]\arrow[u, "\vec t\boxtimes\vec s", dashed,red] &
      \end{tikzcd}
    \end{center}

    The case where $x\in\FV(\vec t)$ is analogous.
  \item \needspace{4em}
    ${ \vcenter{ \infer[^{\mathsf{UnitarySeq}}] {\Gamma,\Xi,x:A\vdash\vec t;\vec s:\sharp B} { \Gamma\vdash\vec t:\sharp\mathbb U & \Xi,x:A\vdash\vec s:\sharp B}}}$
    \begin{center}
      \begin{tikzcd}[labels=description,column sep=2cm,row sep=1.2cm
        , execute at end picture={
          \path (\tikzcdmatrixname-1-2) -- (\tikzcdmatrixname-2-2) coordinate[pos=0.5](aux)
          (\tikzcdmatrixname-1-1) -- (aux) node[midway,blue]{\small (Def)};
          \path (\tikzcdmatrixname-3-2) -- (\tikzcdmatrixname-1-3) node[midway,blue]{\small (Def)};
          \path (\tikzcdmatrixname-1-1) -- (\tikzcdmatrixname-3-2) node[midway,blue]{\small (IH)};
      }
      ]
        \Gamma\boxtimes\Xi\boxtimes\Delta \arrow[rdd, "\Id\boxtimes v"', bend right,sloped] \arrow[rr, "{\vec t;\vec s[x:=v]}", bend left=15] \arrow[dr, "{\vec t\boxtimes\vec s[x:=v]}", dashed,red,sloped] & {\color{red}\Forg{(I\otimes\Span B)}}\ar[r,"\Forg{\lambda_\otimes^{-1}}",sloped,dashed,red] & \Forg{\Span{B}} \\
        & {\color{red}(\Forg{\Span{1}})\boxtimes(\Forg{\Span{B}})} \arrow[u, "n", dashed,red] &   \\
        & \Gamma\boxtimes\Xi\boxtimes A \arrow[ruu, "\vec t;\vec s"', bend right,sloped]\arrow[u, "\vec t\boxtimes\vec s", dashed,red] &
      \end{tikzcd}
    \end{center}

    The case where $x\in\FV(\vec t)$ is analogous.
  \item ${\vcenter{\infer[^{\mathsf{Pair}}]{\Gamma,\Xi,x:A\vdash(\vec u,\vec w):B\times C}{\Gamma\vdash\vec u:B & \Xi,x:A\vdash\vec w:C}}}$
    \begin{center}
      \begin{tikzcd}[row sep=3cm,labels=description,execute at end picture={ \path
          (\tikzcdmatrixname-1-1) -- (\tikzcdmatrixname-1-3)
          coordinate[pos=0.5](aux) (aux) -- (\tikzcdmatrixname-2-2)
          node[pos=0.15,blue]{\small (Def)} (aux) -- (\tikzcdmatrixname-2-2)
          node[pos=0.7,blue]{\small (IH)}; }]
        \Gamma\boxtimes\Xi\boxtimes\Delta \arrow[rd, "\Id\boxtimes v"', bend right,sloped]
        \arrow[rr, "{(\vec u,\vec w[x:=v])}"] \arrow[rr, "{\vec
          u\boxtimes\vec w[x:=v]}", dashed, bend right,red] &
        & B\boxtimes C \\
        & \Gamma\boxtimes\Xi\boxtimes A \arrow[ru, "\vec u\boxtimes\vec w"', bend right,sloped] &
      \end{tikzcd}
    \end{center}
    The case where $x\in\FV(\vec u)$ is analogous.

  \item
    ${\vcenter{\infer[^{\mathsf{PureLet}}]
        {\Gamma,x:A,\Xi\vdash\mathsf{let}\ (y,z)=\vec t\ \mathsf{in}\ \vec s:D}
        { \Gamma,x:A\vdash\vec t:B\times C & \Xi,y:B,z:C\vdash\vec s:D } }}$

      \begin{tikzcd}
        [labels=description,ampersand replacement = \&,column sep=15mm,execute
        at end picture={
          \path (\tikzcdmatrixname-1-1) -- (\tikzcdmatrixname-1-3) coordinate[pos=0.5](aux1)
          (\tikzcdmatrixname-2-1) -- (\tikzcdmatrixname-2-3) coordinate[pos=0.5](aux2)
          (aux1) -- (aux2) node[midway,blue]{\small (Def)};
          \path (\tikzcdmatrixname-4-2) -- (\tikzcdmatrixname-1-3) node[pos=0.1,blue]{\small (Def)}; \path
          (\tikzcdmatrixname-2-1) -- (\tikzcdmatrixname-3-3) coordinate[pos=0.5](aux)
          (aux) -- (\tikzcdmatrixname-3-2) node[midway,yshift=-2mm,blue]{\small (IH)}; \path (\tikzcdmatrixname-1-1) --
          (\tikzcdmatrixname-4-2) node[pos=0.6,blue,sloped,yshift=-9mm]{\small (Nat. of
            swap)}; }]
        \Gamma\boxtimes\Xi\boxtimes\Delta \arrow[rddd, "\Id\boxtimes v",sloped,out=200,in=180] \arrow[d, "swap",dashed,red] \arrow[rr, "{\mathsf{let}\ (y,z)=\vec t{[x:=v]}\ \mathsf{in}\ \vec s}"] \& \& D \\
        {\color{red}\Gamma\boxtimes\Delta\boxtimes\Xi} \arrow[rrd, "{\vec t[x:=v]\boxtimes\eta^{B\boxtimes C}}",sloped,dashed,red] \arrow[rd,"\Id\boxtimes v\boxtimes\Id",sloped,swap,dashed,red] \& \& {{\color{red}(B\boxtimes C)\boxtimes[B\boxtimes C,D]}} \arrow[u, "\varepsilon",sloped,dashed,red] \&   \\
        \& {\color{red}\Gamma\boxtimes A\boxtimes\Xi} \arrow[r, "\vec
        t\boxtimes\eta^{B\boxtimes C}",dashed,red,swap] \&
        {\color{red}(B\boxtimes C)\boxtimes[B\boxtimes C,\Xi\boxtimes B\boxtimes
          C]} \arrow[u,"{\Id\boxtimes[B\boxtimes C,\vec s]}",dashed,red] \&   \\
        \& \Gamma\boxtimes\Xi\boxtimes A \arrow[u, "swap",dashed,red]
        \arrow[ruuu,"{\mathsf{let}\ (y,z)=\vec t\ \mathsf{in}\ \vec s}",near start,sloped,out=0,in=330,swap,looseness=2.7] \& \&
      \end{tikzcd}

  \item
    \needspace{4em}
    ${\vcenter{\infer[^{\mathsf{PureLet}}]
        {\Gamma,\Xi,x:A\vdash\mathsf{let}\ (y,z)=\vec t\ \mathsf{in}\ \vec s:D}
        { \Gamma\vdash\vec t:B\times C & \Xi,x:A,y:B,z:C\vdash\vec s:D } }}$

	\hspace{-2.5cm}\begin{tikzcd}[labels=description,column sep=0mm,ampersand replacement = \&,execute at end picture={
        \path
        (\tikzcdmatrixname-2-1) -- (\tikzcdmatrixname-3-1) coordinate[pos=0.5](aux1)
        (\tikzcdmatrixname-1-2) -- (\tikzcdmatrixname-2-2) coordinate[pos=0.5](aux2)
        (aux1) -- (aux2) node[pos=0.3,blue,sloped]{\small (Nat.~of $\eta^{B\boxtimes C}$)}
        (\tikzcdmatrixname-2-2) -- (\tikzcdmatrixname-4-2) coordinate[pos=0.5](aux)
        (aux) -- (\tikzcdmatrixname-3-1) node[midway,blue]{\small (IH)}
        (aux) -- (\tikzcdmatrixname-3-1) node[pos=-0.75,blue]{\small (Def)};
        \path (\tikzcdmatrixname-2-2) -- (\tikzcdmatrixname-5-2) coordinate[pos=0.5](aux2)
        (aux1) -- (aux2) node[midway,blue,sloped,yshift=-1.5cm]{\small (Def)};
        }]
        \& \Gamma\boxtimes\Xi\boxtimes\Delta \arrow[ld,"\Id\boxtimes v",sloped]\arrow[dddd,"{{\mathsf{let}~(y,z)=\vec{t}~\mathsf{in}~\vec{s}{[x:=v]}}}",out=-5,in=15,looseness=2,sloped,swap,pos=0.8]\arrow[d,"\vec t\boxtimes\eta^{B\boxtimes C}",dashed,red,swap] \\
        \Gamma\boxtimes\Xi\boxtimes
        A\arrow[rddd,"{{\mathsf{let}~(y,z)=\vec{t}~\mathsf{in}~\vec{s}}}",swap,sloped,controls={(-10,0)
          and (-9,0)},pos=0.7]\arrow[d,"\vec t\boxtimes\eta^{B\boxtimes
          C}",swap,dashed,red] \& {\color{red}(B\boxtimes C)\boxtimes[B\boxtimes C,\Xi\boxtimes\Delta\boxtimes B\boxtimes
          C]}\arrow[dd, "{\Id\boxtimes[B\boxtimes C,\vec s[x:=v]]}", dashed,red]
        \arrow[ld,"{\Id\boxtimes[B\boxtimes C,\Id\boxtimes v\boxtimes\Id]}"',sloped,dashed,red] \\
        {\color{red}(B\boxtimes C)\boxtimes[B\boxtimes C,\Xi\boxtimes
          A\boxtimes B\boxtimes C]}\arrow[rd,"{\Id\boxtimes[B\boxtimes C,\vec s]}",sloped,dashed,red] \&\\
        \& {\color{red}(B\boxtimes C)\boxtimes[B\boxtimes C,D]} \arrow[d, "\varepsilon", dashed,red] \\
        \& D
      \end{tikzcd}

  \item $\vcenter{ \infer[^{\mathsf{UnitaryLet}}]
      {\Gamma,x:A,\Xi\vdash\mathsf{let}~(y,z)=\vec{t}~\mathsf{in}~\vec{s}:\sharp
        D} { \Gamma,x:A\vdash\vec t:B\otimes C & \Xi,y:\sharp B,z:\sharp
        C\vdash\vec s:\sharp D } }$

    \hspace{-1cm}\scalebox{.83}{
      \begin{tikzcd}[labels=description,row sep=1cm,ampersand replacement = \&,execute at end picture={ \path
          (\tikzcdmatrixname-1-1) -- (\tikzcdmatrixname-2-1)
          coordinate[pos=0.5](aux1) (\tikzcdmatrixname-1-3) --
          (\tikzcdmatrixname-2-3) coordinate[pos=0.5](aux2) (aux1) -- (aux2)
          node[midway,blue]{\small (Naturality of swap)} (\tikzcdmatrixname-2-1)
          -- (\tikzcdmatrixname-2-3) coordinate[pos=0.5](aux3) (aux3) --
          (\tikzcdmatrixname-3-2) node[midway,blue]{\small (IH)}
          (\tikzcdmatrixname-2-1) -- (\tikzcdmatrixname-8-2)
          node[pos=0.4,blue]{\small (Def)} (\tikzcdmatrixname-2-3) --
          (\tikzcdmatrixname-8-2) node[pos=0.4,blue]{\small (Def)} ; }]
        \Gamma\boxtimes\Xi\boxtimes A \arrow[d, "swap", dashed,red] \arrow[rddddddd,"{\mathsf{let}~(y,z)=\vec{t}~\mathsf{in}~\vec{s}}"',out=220,in=150,sloped] \&\& \Gamma\boxtimes\Xi\boxtimes\Delta \arrow[ll, "\Id\boxtimes v"'] \arrow[d, "swap"', dashed,red] \arrow[lddddddd, "{\mathsf{let}~(y,z)=\vec{t}[x:=v]~\mathsf{in}~\vec{s}}",out=-40,in=30,sloped,swap] \\
        {\color{red}\Gamma\boxtimes A\boxtimes\Xi} \arrow[rd, "\vec{t}\boxtimes\eta^{\Forg{\Span{B}}\boxtimes \Forg{\Span{C}}}",sloped,dashed,red] \&\& {\color{red}\Gamma\boxtimes\Delta\boxtimes\Xi} \arrow[ll, "\Id\boxtimes v\boxtimes\Id",swap,dashed,red] \arrow[ld, "{\vec{t}[x:=v]\boxtimes\eta^{\Forg{\Span{B}}\boxtimes \Forg{\Span{C}}}}",sloped,dashed,red] \\
        \& {\color{red}{\Forg{\Span{(B\boxtimes C)}}\boxtimes[\Forg{\Span{B}}\boxtimes
            \Forg{\Span{C}},\Xi\boxtimes \Forg{\Span{B}}\boxtimes
            \Forg{\Span{C}}]}} \arrow[d, "{\Id\boxtimes[\Forg{\Span
            B}\boxtimes\Forg{\Span C},\vec{s}]}", dashed,red] \&\\
        \& {\color{red}\Forg{\Span{(B\boxtimes C)}}\boxtimes[\Forg{\Span{B}}\boxtimes \Forg{\Span{C}},\Forg{\Span{D}}]} \arrow[d, "\Forg{m}\,\circ\, n\,\circ\, (\Id\boxtimes\eta)", dashed,red] \&\\
        \& {\color{red}i\sharp(B\boxtimes C\boxtimes[\Forg{\Span{B}}\boxtimes \Forg{\Span{C}},\Forg{\Span{D}}])} \arrow[d, "{\Forg{\Span{(\eta\boxtimes\eta\boxtimes\Id)}}}", dashed,red] \&\\
        \& {\color{red}\Forg{\Span{(\Forg{\Span{B}}\boxtimes \Forg{\Span{C}}\boxtimes[\Forg{\Span{B}}\boxtimes \Forg{\Span{C}},\Forg{\Span{D}}])}}} \arrow[d, "\Forg{\Span{\varepsilon}}", dashed,red] \&\\
        \& {\color{red}\Forg{\Span{D}} \Forg{\Span{D}}} \arrow[d, "\mu", dashed,red] \& \\
        \& \Forg{\Span{D}} \&
     \end{tikzcd}
   }

 \item
    \needspace{4em}
   $\vcenter{ \infer[^{\mathsf{UnitaryLet}}]
     {\Gamma,\Xi,x:A\vdash\mathsf{let}~(y,z)=\vec{t}~\mathsf{in}~\vec{s}:\sharp
       D} { \Gamma\vdash\vec t:B\otimes C & \Xi,x:A,y:\sharp B,z:\sharp
       C\vdash\vec s:\sharp D } }$

  \hspace{-1cm}\scalebox{.83}{
     \begin{tikzcd}[labels=description,row sep=1cm,column sep=4mm,ampersand replacement =
       \&,execute at end picture={ \path (\tikzcdmatrixname-1-1) --
         (\tikzcdmatrixname-2-1) coordinate[pos=0.5](aux1)
         (\tikzcdmatrixname-1-3) -- (\tikzcdmatrixname-2-3)
         coordinate[pos=0.5](aux2) (aux1) -- (aux2) node[midway,blue]{\small
           (Naturality of $\eta^{\Forg{\Span{B}}\boxtimes \Forg{\Span{C}}}$)}
         (\tikzcdmatrixname-2-1) -- (\tikzcdmatrixname-2-3)
         coordinate[pos=0.5](aux3) (aux3) -- (\tikzcdmatrixname-3-2)
         node[midway,blue]{\small (IH)} (\tikzcdmatrixname-2-1) --
         (\tikzcdmatrixname-7-2)
         node[midway,blue,xshift=-1.2cm,yshift=1.2cm]{\small (Def)}
         (\tikzcdmatrixname-2-3) -- (\tikzcdmatrixname-7-2)
         node[midway,blue,xshift=1.2cm,yshift=1.2cm]{\small (Def)} ; }]
       \Gamma\boxtimes\Xi\boxtimes A \arrow[d, "\vec{t}\boxtimes\eta^{\Forg{\Span{B}}\boxtimes \Forg{\Span{C}}}", dashed,red] \arrow[rdddddd,"{\mathsf{let}~(y,z)=\vec{t}~\mathsf{in}~\vec{s}}"',out=200,in=150,sloped] \&\& \Gamma\boxtimes\Xi\boxtimes\Delta \arrow[ll, "\Id\boxtimes v"'] \arrow[d, "\vec{t}\boxtimes\eta^{\Forg{\Span{B}}\boxtimes \Forg{\Span{C}}}"', dashed,red] \arrow[ldddddd, "{\mathsf{let}~(y,z)=\vec{t}~\mathsf{in}~\vec{s}[x:=v]}",out=-20,in=30,sloped,swap] \\
       \parbox{3cm}{$\color{red}\Forg{\Span{(B\boxtimes C)}}\boxtimes[\Forg{\Span{B}}\boxtimes \Forg{\Span{C}},\Xi\boxtimes A\boxtimes \Forg{\Span{B}}\boxtimes \Forg{\Span{C}}]$} \arrow[rd, "\Id\boxtimes{[\Forg{\Span B}\boxtimes\Forg{\Span C},\vec{s}]}",sloped,dashed,red] \&\& \parbox{3cm}{$\color{red}\Forg{\Span{(B\boxtimes C)}}\boxtimes[\Forg{\Span{B}}\boxtimes \Forg{\Span{C}},\Xi\boxtimes \Delta\boxtimes \Forg{\Span{B}}\boxtimes \Forg{\Span{C}}]$} \arrow[ll, "\Id\boxtimes v\boxtimes\Id",swap,dashed,red] \arrow[ld, "\Id\boxtimes{[\Forg{\Span B}\boxtimes\Forg{\Span C},\vec{s}[x:=v]]}",sloped,dashed,red] \\
       \& {\color{red}\Forg{\Span{(B\boxtimes C)}}\boxtimes[\Forg{\Span{B}}\boxtimes \Forg{\Span{C}},\Forg{\Span{D}}]} \arrow[d, "\Forg{m}\,\circ\, n\,\circ\, (\Id\boxtimes\eta)", dashed,red] \&\\
       \& {\color{red}\Forg{\Span{(B\boxtimes C\boxtimes[\Forg{\Span{B}}\boxtimes \Forg{\Span{C}},\Forg{\Span{D}}])}}} \arrow[d, "{\Forg{\Span{(\eta\boxtimes\eta,\Id)}}}", dashed,red] \&\\
       \& {\color{red}\Forg{\Span{(\Forg{\Span{B}}\boxtimes
             \Forg{\Span{C}}\boxtimes[\Forg{\Span{B}}\boxtimes
             \Forg{\Span{C}},\Forg{\Span{D}}])}}} \arrow[d, "{\Forg{\Span\varepsilon}}", dashed,red] \&\\
       \& {\color{red}\Forg{\Span{D}} \Forg{\Span{D}}} \arrow[d, "\mu", dashed,red] \& \\
       \& \Forg{\Span{D}} \&
     \end{tikzcd}
   }

 \item ${\vcenter{\infer[^{\mathsf{InL}}]{\Gamma,x:A\vdash\mathsf{inl}(\vec v):B+C}{\Gamma,x:A\vdash\vec v:B}}}$
   \begin{center}
     \begin{tikzcd}[labels=description,column sep=3cm,execute at end picture={ \path
         (\tikzcdmatrixname-1-1) -- (\tikzcdmatrixname-1-3)
         coordinate[pos=0.5](aux) (aux) -- (\tikzcdmatrixname-2-2)
         node[midway,blue]{\small (Def)}; \path (\tikzcdmatrixname-3-2) --
         (\tikzcdmatrixname-1-3) coordinate[pos=0.5](aux) (aux) --
         (\tikzcdmatrixname-2-2) node[midway,blue]{\small (Def)}; \path
         (\tikzcdmatrixname-1-1) -- (\tikzcdmatrixname-3-2)
         coordinate[pos=0.5](aux) (aux) -- (\tikzcdmatrixname-2-2)
         node[midway,blue]{\small (IH)}; }]
       \Gamma\boxtimes\Delta \arrow[rdd,sloped,"\Id\boxtimes v"'] \arrow[rr, "{\vec{w}[x:=v]^{B+C}}"] \arrow[rd,sloped,"{\vec{w}[x:=v]}", dashed,red] & & B+C \\
       &{\color{red}B} \arrow[ru,sloped,"i_1", dashed,red]  &     \\
       & \Gamma\boxtimes A \arrow[ruu,sloped,"\vec{w}^{B+C}"'] \arrow[u, "\vec{w}",
       dashed,red] &
     \end{tikzcd}
   \end{center}
 \item ${\vcenter{\infer[^{\mathsf{InR}}]{\Gamma,x:A\vdash\mathsf{inr}(\vec v):B+C}{\Gamma,x:A\vdash\vec v:C}}}$\quad Analogous to previous case.

 \item
    \needspace{4em}
   ${ \vcenter{ \infer[^{\mathsf{PureMatch}}]
       {\Gamma,x:A,\Xi\vdash\mathsf{match}~\vec{t}~\{\mathsf{inl}(x_1)\mapsto\vec{v}_1\mid\mathsf{inr}(x_2)\mapsto\vec{v}_2\}:D}
       { \Gamma,x:A\vdash\vec{t}:B+C & \Xi\vdash(x_1: B\vdash\vec{v}_1\perp
         x_2:C\vdash\vec{v}_2):D } } }$
   \begin{center}
     \begin{tikzcd}[labels=description,column sep=1.5cm,row sep=1cm,execute at end picture={ \path (\tikzcdmatrixname-1-1) --
         (\tikzcdmatrixname-2-1) coordinate[pos=0.5](aux1)
         (\tikzcdmatrixname-1-3) -- (\tikzcdmatrixname-2-3)
         coordinate[pos=0.5](aux2) (aux1) -- (aux2) node[midway,blue]{\small
           (Def)} (\tikzcdmatrixname-2-1) -- (\tikzcdmatrixname-2-3)
         coordinate[pos=0.5](aux3) (aux3) -- (\tikzcdmatrixname-3-2)
         node[midway,blue]{\small (IH)} (aux3) -- (\tikzcdmatrixname-3-2)
         node[midway,blue,yshift=-1cm,xshift=-3cm]{\small (Naturality of swap)}
         (aux3) -- (\tikzcdmatrixname-3-2)
         node[midway,blue,yshift=-1cm,xshift=3cm]{\small (Def)}; }]
       \Gamma\boxtimes\Xi\boxtimes\Delta \arrow[rddd, "\Id\boxtimes v"',sloped,out=200,in=180,looseness=2] \arrow[rr, "{\mathsf{match}~\vec{t}[x:=v]~\{\mathsf{inl}(x_1)\mapsto\vec{v}_1\mid\mathsf{inr}(x_2)\mapsto\vec{v}_2\}}"] \arrow[d, "swap", dashed,red] & & D\\
       \color{red}\Gamma\boxtimes\Delta\boxtimes\Xi \arrow[rd,sloped,swap,"\Id\boxtimes v\boxtimes\Id", dashed,red] \arrow[rr, "{\vec{t}[x:=v]\boxtimes\Id}", dashed,red] & & \color{red}(B+C)\boxtimes\Xi \arrow[u, "{[\vec{v}_1,\vec{v}_2]_1}", dashed,red] \\
       & \color{red}\Gamma\boxtimes A\boxtimes\Xi \arrow[ru,sloped,swap,"\vec{t}\boxtimes\Id", dashed,red] & \\
       & \Gamma\boxtimes\Xi\boxtimes A \arrow[ruuu,near start,"\mathsf{match}~\vec{t}~\{\mathsf{inl}(x_1)\mapsto\vec{v}_1\mid\mathsf{inr}(x_2)\mapsto\vec{v}_2\}"',out=0,in=-20,looseness=2,sloped]
       \arrow[u, "swap", dashed,red] &
     \end{tikzcd}
   \end{center}

 \item\needspace{4em}
 ${ \vcenter{ \infer[^{\mathsf{PureMatch}}]
       {\Gamma,\Xi,x:A\vdash\mathsf{match}~\vec{t}~\{\mathsf{inl}(x_1)\mapsto\vec{v}_1\mid\mathsf{inr}(x_2)\mapsto\vec{v}_2\}:D}
       { \Gamma\vdash\vec{t}:B+C & \Xi,x:A\vdash(x_1: B\vdash\vec{v}_1\perp
         x_2:C\vdash\vec{v}_2):D } } }$
   \begin{center}
     \begin{tikzcd}[labels=description,row sep=1cm,column sep=3cm,execute at end picture={ \path
         (\tikzcdmatrixname-1-1) -- (\tikzcdmatrixname-1-3)
         coordinate[pos=0.5](aux) (aux) -- (\tikzcdmatrixname-2-2)
         node[midway,blue]{\small (Def)} (\tikzcdmatrixname-1-1) --
         (\tikzcdmatrixname-4-2) node[midway,blue,xshift=2mm]{\small
           (Functoriality of $\boxtimes$)} (\tikzcdmatrixname-4-2) --
         (\tikzcdmatrixname-1-3) node[midway,blue,xshift=-7mm]{\small (IH)}
         (\tikzcdmatrixname-4-2) -- (\tikzcdmatrixname-1-3)
         node[midway,blue,xshift=-15mm,yshift=-15mm]{\small (Def)} ; }]
       \Gamma\boxtimes\Xi\boxtimes\Delta \arrow[rddd,sloped,bend right,"\Id\boxtimes v"'] \arrow[rr,"{\mathsf{match}~\vec{t}~\{\mathsf{inl}(x_1)\mapsto\vec{v}_1[x:=v]\mid\mathsf{inr}(x_2)\mapsto\vec{v}_2[x:=v]\}}"] \arrow[rd,sloped,swap,"\vec t\boxtimes\Id", dashed,red] & & D \\
       & \color{red}(B+C)\boxtimes\Xi\boxtimes\Delta \arrow[d, "\Id\boxtimes v", dashed,red] \arrow[ru, "{[\vec{v}_1[x:=v],\vec{v}_2[x:=v]]_1}",sloped,swap,dashed,red]&   \\
       & \color{red}(B+C)\boxtimes\Xi\boxtimes A \arrow[ruu,
       "{[\vec{v}_1,\vec{v}_2]_1}",sloped,out=0,in=225,dashed,red] &   \\
       & \Gamma\boxtimes\Xi\boxtimes A
       \arrow[ruuu,"\mathsf{match}~\vec{t}~\{\mathsf{inl}(x_1)\mapsto\vec{v}_1\mid\mathsf{inr}(x_2)\mapsto\vec{v}_2\}"',sloped,bend
       right]\arrow[u, "\vec t\boxtimes\Id", dashed,red] &
     \end{tikzcd}
   \end{center}

 \item\needspace{4em}
 ${ \vcenter{ \infer[^{\mathsf{UnitaryMatch}}]
       {\Gamma,x:A,\Xi\vdash\mathsf{match}~\vec{t}~\{\mathsf{inl}(x_1)\mapsto\vec{v}_1\mid\mathsf{inr}(x_2)\mapsto\vec{v}_2\}:\sharp
         D} { \Gamma,x:A\vdash\vec{t}:B\oplus C & \Xi\vdash(x_1: \sharp
         B\vdash\vec{v}_1\perp x_2:\sharp C\vdash\vec{v}_2):\sharp D } } }$
   \begin{center}
     \begin{tikzcd}[labels=description,execute at end picture={ \path (\tikzcdmatrixname-1-1) --
         (\tikzcdmatrixname-3-1) coordinate[pos=0.5](aux1)
         (\tikzcdmatrixname-1-3) -- (\tikzcdmatrixname-3-3)
         coordinate[pos=0.5](aux2) (aux1) -- (aux2) node[midway,blue]{\small
           (Def)} (\tikzcdmatrixname-6-2) -- (\tikzcdmatrixname-1-3)
         node[midway,blue,yshift=-2cm]{\small (Def)} (\tikzcdmatrixname-6-2) --
         (\tikzcdmatrixname-1-1) node[midway,blue,yshift=-1cm]{\small (IH)}
         (\tikzcdmatrixname-6-2) -- (\tikzcdmatrixname-3-1)
         node[midway,blue,sloped,yshift=-1.2cm,xshift=-5mm]{\small (Naturality
           of swap)} ; }]
       \Gamma\boxtimes\Xi\boxtimes\Delta \arrow[rddddd,"\Id\boxtimes v"',sloped,out=195,in=180] \arrow[dd, "swap"', dashed,red] \arrow[rr, "{\mathsf{match}~\vec{t}[x:=v]~\{\mathsf{inl}(x_1)\mapsto\vec{v}_1\mid\mathsf{inl}(x_2)\mapsto\vec{v}_2\}}"] && \Forg{\Span{D}} \\
       && \color{red}\Forg{\Span{((\Forg{\Span{B}}+\Forg{\Span{C}})}}\boxtimes\Xi) \arrow[u, "{\Forg{\Span{[\vec{v}_1,\vec{v}_2]_1}}}", dashed,red] \\
       \color{red}\Gamma\boxtimes\Delta\boxtimes\Xi \arrow[rdd, "\Id\boxtimes
       v\boxtimes\Id"',out=-60,in=180,sloped,dashed,red]
       \arrow[rd,sloped,"{\vec{t}[x:=v]\boxtimes\Id}", dashed,red] &&\color{red}
       \Forg{\Span{(\Forg{\Span{B}}+\Forg{\Span{C}})}}\boxtimes \Forg{\Span{\Xi}} \arrow[u,
       "\Forg m\,\circ\, n", dashed,red] \\
       &\color{red} \Forg{\Span{(B+C)}}\boxtimes\Xi \arrow[ru,sloped,"\Forg{\Span{(\eta+\eta)}}\boxtimes\eta", dashed,red] & \\
       &\color{red} \Gamma\boxtimes A\boxtimes\Xi \arrow[u, "\vec{t}\boxtimes\Id", dashed,red] &\\
       & \Gamma\boxtimes\Xi\boxtimes A \arrow[u, "swap", dashed,red]
       \arrow[ruuuuu,near start,"\mathsf{match}~\vec{t}~\{\mathsf{inl}(x_1)\mapsto\vec{v}_1\mid\mathsf{inl}(x_2)\mapsto\vec{v}_2\}"',out=0,in=-10,looseness=1.5,sloped] &
     \end{tikzcd}
   \end{center}

 \item\needspace{4em}
 ${ \vcenter{ \infer[^{\mathsf{UnitaryMatch}}]
       {\Gamma,\Xi,x:A\vdash\mathsf{match}~\vec{t}~\{\mathsf{inl}(x_1)\mapsto\vec{v}_1\mid\mathsf{inr}(x_2)\mapsto\vec{v}_2\}:\sharp
         D} { \Gamma\vdash\vec{t}:B\oplus C & \Xi,x:A\vdash(x_1: \sharp
         B\vdash\vec{v}_1\perp x_2:\sharp C\vdash\vec{v}_2):\sharp D } } }$

\newcommand\texto{\ensuremath{\mathsf{match}~\vec{t}~\{\mathsf{inl}(x_1)\mapsto\vec{v}_1\mid\mathsf{inl}(x_2)\mapsto\vec{v}_2\}}}

\hspace{-3cm} \begin{tikzcd}
       [ labels=description, ampersand replacement = \&, column sep=1cm, row
       sep=1.5cm, execute at end picture={ \path (\tikzcdmatrixname-1-1) --
         (\tikzcdmatrixname-2-1) coordinate[pos=0.5](aux1)
         (\tikzcdmatrixname-1-3) -- (\tikzcdmatrixname-2-3)
         coordinate[pos=0.5](aux2) (aux1) -- (aux2) node[midway,blue]{\small
           (Functoriality of $\boxtimes$)}; \path (\tikzcdmatrixname-2-1) --
         (\tikzcdmatrixname-3-1) coordinate[pos=0.5](aux1)
         (\tikzcdmatrixname-2-3) -- (\tikzcdmatrixname-3-3)
         coordinate[pos=0.5](aux2) (aux1) -- (aux2) node[midway,blue]{\small
           (Naturality of $\eta$)}; \path (\tikzcdmatrixname-3-1) --
         (\tikzcdmatrixname-4-1) coordinate[pos=0.5](aux1)
         (\tikzcdmatrixname-3-3) -- (\tikzcdmatrixname-4-3)
         coordinate[pos=0.5](aux2) (aux1) -- (aux2) node[midway,blue]{\small
           (Naturality of $n$ and $m$)}; \path (\tikzcdmatrixname-4-1) --
         (\tikzcdmatrixname-4-3) coordinate[pos=0.5](aux) (aux) --
         (\tikzcdmatrixname-5-2) node[midway,blue]{\small (IH)} (aux) --
         (\tikzcdmatrixname-5-2) node[midway,blue,xshift=4cm]{\small (Def)}
         (aux) -- (\tikzcdmatrixname-5-2) node[midway,blue,xshift=-4cm]{\small
           (Def)}; }]
       \Gamma\boxtimes\Xi\boxtimes A
       \arrow[rdddd,out=225,in=180,looseness=2.2,near end,
       "\mathsf{match}~\vec{t}~\{\mathsf{inl}(x_1)\mapsto\vec{v}_1\mid\mathsf{inl}(x_2)\mapsto\vec{v}_2\}",
       sloped] \arrow[d, "\vec{t}\boxtimes\Id", dashed,red] \& \& \Gamma\boxtimes\Xi\boxtimes\Delta \arrow[ll, "\Id\boxtimes v"'] \arrow[ldddd, "{\mathsf{match}~\vec{t}~\{\mathsf{inl}(x_1)\mapsto\vec{v}_1[x:=v]\mid\mathsf{inl}(x_2)\mapsto\vec{v}_2[x:=v]\}}",out=-45,in=0,sloped,swap,looseness=2.2,near end] \arrow[d, "\vec{t}\boxtimes\Id", dashed,red] \\
       \color{red}\Forg{\Span{(B+C)}}\boxtimes\Xi\boxtimes A \arrow[d, "\Forg{\Span{(\eta+\eta)}}\boxtimes\eta", dashed,red] \& \&\color{red} \Forg{\Span{(B+C)}}\boxtimes\Xi\boxtimes\Delta \arrow[d, "\Forg{\Span{(\eta+\eta)}}\boxtimes\eta", dashed,red] \arrow[ll, "\Id\boxtimes v", dashed,red] \\
       \color{red}\Forg{\Span{(\Forg{\Span{B}}+\Forg{\Span{C}})}}\boxtimes \Forg{\Span{(\Xi\boxtimes A)}} \arrow[d, "\Forg{m}\circ\, n", dashed,red] \& \&\color{red} \Forg{\Span{(\Forg{\Span{B}}+\Forg{\Span{C}})}}\boxtimes \Forg{\Span{(\Xi\boxtimes\Delta)}} \arrow[d, "\Forg{m}\circ\, n", dashed,red] \arrow[ll, "\Id\boxtimes \Forg{\Span{(\Id\boxtimes v)}}", dashed,red] \\
       \color{red}\Forg{\Span{((\Forg{\Span{B}}+\Forg{\Span{C}})}}\boxtimes\Xi\boxtimes A) \arrow[rd,sloped,"{\Forg{\Span{[\vec{v}_1,\vec{v}_2]_1}}}",dashed,red] \& \& \color{red}\Forg{\Span{((\Forg{\Span{B}}+\Forg{\Span{C}})}}\boxtimes\Xi\boxtimes\Delta) \arrow[ld,sloped,"{\Forg{\Span{[\vec{v}_1[x:=v],\vec{v}_2[x:=v]]_1}}}", dashed,red] \arrow[ll, "\Forg{\Span{(\Id\boxtimes v)}}",swap,dashed,red] \\
       \& \Forg{\Span{D}} \&
     \end{tikzcd}

 \item\needspace{4em}
 ${\vcenter{\infer[^{\mathsf{Weak}}]{\Gamma,x:A,y:B\vdash\vec t:C}{\Gamma,x:A\vdash\vec t:C & B^\flat}}}$
   \begin{center}
     \begin{tikzcd}[labels=description,column sep=4cm,execute at end picture={ \path
         (\tikzcdmatrixname-1-1) -- (\tikzcdmatrixname-3-2)
         coordinate[pos=0.5](aux) (aux) -- (\tikzcdmatrixname-1-2)
         node[midway,blue]{\small (Funct.~of $\boxtimes$)}; \path
         (\tikzcdmatrixname-1-1) -- (\tikzcdmatrixname-1-3)
         coordinate[pos=0.5](aux) (aux) -- (\tikzcdmatrixname-2-2)
         node[midway,blue,yshift=1.5cm]{\small (Def)}; \path
         (\tikzcdmatrixname-2-2) -- (\tikzcdmatrixname-1-3)
         coordinate[pos=0.5](aux) (aux) -- (\tikzcdmatrixname-1-2)
         node[midway,blue]{\small (IH)}; \path (\tikzcdmatrixname-3-2) --
         (\tikzcdmatrixname-1-3) coordinate[pos=0.5](aux) (aux) --
         (\tikzcdmatrixname-2-2) node[midway,blue]{\small (Def)}; }]
       \Gamma\boxtimes B\boxtimes\Delta \arrow[rdd, "\Id\boxtimes v"',sloped] \arrow[rr, "{\vec{t}[x:=v]_{\Gamma\boxtimes B\boxtimes\Delta}}", bend left] \arrow[r, "{\pi'_\Gamma\boxtimes\Id}", dashed,red] & \color{red}\Gamma\boxtimes\Delta \arrow[d, "\Id\boxtimes v", dashed,red] \arrow[r, "{\vec{t}[x:=v]_{\Gamma\boxtimes\Delta}}", dashed,red] & C \\
       & \color{red}\Gamma\boxtimes A \arrow[ru, "\vec{t}_{\Gamma\boxtimes A}",sloped,dashed,red] &   \\
       & \Gamma\boxtimes B\boxtimes A \arrow[ruu, "\vec{t}_{\Gamma\boxtimes B\boxtimes
         A}"',sloped]\arrow[u, "{\pi'_\Gamma\boxtimes\Id}", dashed,red] &
     \end{tikzcd}
   \end{center}

 \item\needspace{4em}
 ${\vcenter{\infer[^{\mathsf{Contr}}]{\Gamma,y:B,x:A\vdash\vec t[z:=y]:C}{\Gamma,y:B,x:A,z:B\vdash\vec t:C & B^\flat}}}$
   \begin{center}
     \begin{tikzcd}[labels=description,column sep=4cm,execute at end picture={
	   \path (\tikzcdmatrixname-1-1) -- (\tikzcdmatrixname-3-2)
	   coordinate[pos=0.5](aux) (aux) -- (\tikzcdmatrixname-1-2)
	   node[midway,blue]{\small (Funct.~of $\boxtimes$)}; \path
	   (\tikzcdmatrixname-1-1) -- (\tikzcdmatrixname-1-3)
	   coordinate[pos=0.5](aux) (aux) -- (\tikzcdmatrixname-2-2)
	   node[midway,blue,yshift=1.5cm]{\small (Def)}; \path
	   (\tikzcdmatrixname-2-2) -- (\tikzcdmatrixname-1-3)
	   coordinate[pos=0.5](aux) (aux) -- (\tikzcdmatrixname-1-2)
	   node[midway,blue,xshift=-7mm]{\small (IH)}; \path
	   (\tikzcdmatrixname-3-2) -- (\tikzcdmatrixname-1-3)
	   coordinate[pos=0.5](aux) (aux) -- (\tikzcdmatrixname-2-2)
	   node[midway,blue,xshift=-2mm]{\small (Def)}; }]
       \Gamma\boxtimes B\boxtimes\Delta \arrow[rdd, "\Id\boxtimes v"',sloped] \arrow[rr, "{\vec{t}[z:=y,x:=v]_{\Gamma\boxtimes B\boxtimes\Delta}}", bend left]
       \arrow[r,"{\Id\boxtimes\delta\boxtimes\Id}", dashed,red] & \color{red}\Gamma\boxtimes B\boxtimes B\boxtimes \Delta \arrow[d, "\Id\boxtimes v",dashed,red] \arrow[r,"{\vec{t}[z:=y,x:=v]_{\Gamma\boxtimes B\boxtimes B\boxtimes \Delta}}", dashed,red] & C \\
       & \color{red}\Gamma\boxtimes B\boxtimes B\boxtimes A
       \arrow[ru,"\vec{t}{[z:=y]}_{\Gamma\boxtimes B\boxtimes B\boxtimes
         A}",sloped,dashed,red,near start] &   \\
       & \Gamma\boxtimes B\boxtimes A \arrow[ruu, "\vec{t}{[z:=y]}_{\Gamma\boxtimes
         B\boxtimes A}"',sloped]\arrow[u, "{\Id\boxtimes\delta\boxtimes\Id}", dashed,red] &
     \end{tikzcd}
   \end{center}

 \item\needspace{4em}
 ${
     \vcenter{
       \infer[^{\mathsf{Sup}}]
       {\vdash\sum_{j=1}^m\alpha_j\cdot\vec v_j:\sharp B}
       {
         \text{\scriptsize $(k\neq h)$}
         &
         \vdash(\vec v_k\perp\vec v_h):B
         &
         \sum_{j=1}^m|\alpha_j|^2=1
         &
         m\geq 1
         &
         B\neq C\rightarrow D
       }
     }
   }$

   This case does not follow the hypothesis of the lemma.
   \qedhere
 \end{itemize}
\end{proof}

\xrecap{Theorem}{Soundness}{thm:soundness}{ If $\Gamma\vdash t:A$, and $t\lra
   r$, then $\sem{\Gamma\vdash t:A} = \sem{\Gamma\vdash r:A}$. }
 \begin{proof}
  By induction on the rewrite relation, using the first derivable type for each
  term.

  First we check the congruence rules, since the rewrite relation is defined
  modulo such congruence. Notice that the $\s{Unitary}$ rules allow to
  type term distributions. However, at some point, such distributions where typed
  by the rule $\s{Sup}$, which only occur on closed values. It is easy to check
  that we can always use the congruence at the level of the closed values, and
  inherit such form in the final term distributions. Therefore, we only check the
  congruence rules for closed values.
   \begin{itemize}
   \item $\vec v_1+\vec v_2\equiv\vec v_2+\vec v_1$ and $(\vec v_1+\vec v_2)+\vec v_3\equiv\vec v_1+(\vec v_2+\vec v_3)$ follow from the commutativity and
     associativity of $+$ in $\VecV$.
  \item $1\cdot \vec v\equiv\vec v$.  We have
    \[
      \vcenter
      {
        \infer[{}^{\mathsf{Sup}}]{\vdash 1\cdot\vec v:\sharp A}
        {
          \vdash\vec v:A
        }
      }
      \qquad\textrm{and}\qquad
      \vcenter{
        \infer[{}^{\leq}]{\vdash\vec v:\sharp A}{\vdash\vec v:A}
      }
    \]
    Then,
    
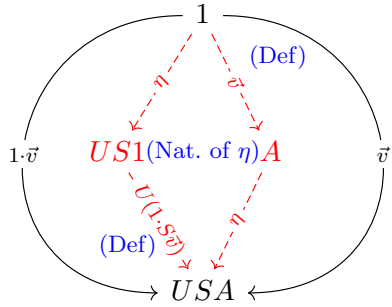
\begin{center}
      \begin{tikzcd}[labels=description,column sep=0cm,row sep=1.3cm,execute at end picture={
          \path
          (\tikzcdmatrixname-1-2) -- (\tikzcdmatrixname-3-2) node[midway,blue]{\footnotesize (Nat.~of $\eta$)}
          (\tikzcdmatrixname-1-2) -- (\tikzcdmatrixname-3-2) node[pos=0.1,xshift=1cm,blue]{\footnotesize (Def)}
          (\tikzcdmatrixname-1-2) -- (\tikzcdmatrixname-3-2) node[pos=0.9,xshift=-1cm,blue]{\footnotesize (Def)}
          ;
        }]
        & 1 \arrow[ld, "\eta" ,sloped,red, dashed] \arrow[dd, "1\cdot\vec v" description,out=180,in=180,looseness=1.8] \arrow[rd, "\vec v" ,sloped,red, dashed] \arrow[dd, "\vec v" description,out=0,in=0,looseness=1.8] & \\
        {\color{red}\Forg{\Span 1}} \arrow[rd, "\Forg{(1\cdot\Span{\vec v})}" ,sloped,red, dashed] & & {\color{red}A} \arrow[ld, "\eta" ,sloped,red, dashed] \\
        & \Forg{\Span A} &
      \end{tikzcd}
    \end{center}

  \item $\alpha\cdot(\beta\cdot\vec v)\equiv\alpha\beta\cdot\vec v$.  We have
    \[
      \vcenter
      {
        \infer[^{\s{Sup}}]{\vdash \alpha\cdot(\beta\cdot\vec v):\sharp\sharp A}
        {
          \infer[^{\s{Sup}}]{\vdash\beta\cdot\vec v:\sharp A}{\vdash\vec v:A}
        }
      }
      \qquad\textrm{and}\qquad
      \vcenter{
        \infer[^{\leq}]{\vdash\alpha\beta\cdot\vec v:\sharp\sharp A}
        {
        \infer[^{\s{Sup}}]{\vdash\alpha\beta\cdot\vec v:\sharp A}
        {\vdash\vec v:A}
        }
      }
    \]
    Then,

    
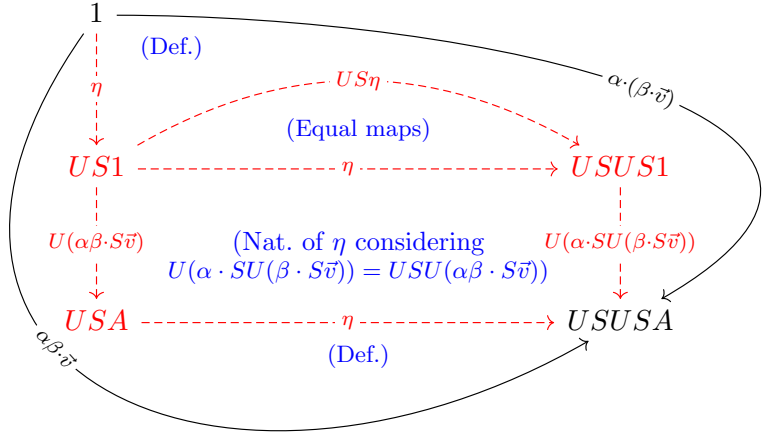
\begin{center}
      \begin{tikzcd}[labels=description,column sep=5.5cm,row sep=1.5cm,
         execute at end picture={
         \path (\tikzcdmatrixname-2-1) -- (\tikzcdmatrixname-3-1) coordinate[pos=0.5](aux1)
           (\tikzcdmatrixname-2-2) -- (\tikzcdmatrixname-3-2) coordinate[pos=0.5](aux2)
           (aux1) -- (aux2) node[midway,blue]{\small(Nat.~of $\eta$ considering}
           (aux1) -- (aux2) node[midway,blue,yshift=-1em]{\footnotesize
             $\Forg{(\alpha\cdot \Span{\Forg{(\beta\cdot \Span{\vec v})}})}
             =
             \Forg{\Span{\Forg{(\alpha\beta\cdot \Span{\vec v})}}}$)
             }
           (aux1) -- (aux2) node[midway,blue,yshift=1.5cm]{\footnotesize (Equal maps)}
           (\tikzcdmatrixname-1-1) -- (aux1) node[midway,blue,xshift=1cm,yshift=1.2cm]{\footnotesize (Def.)}
           (aux1) -- (aux2) node[midway,blue,yshift=-1.5cm]{\footnotesize (Def.)}
           ;}
        ]
        1 \arrow[d, "\eta", red, dashed] \arrow[rdd, "\alpha\cdot(\beta\cdot\vec v)", out=0,in=30,looseness=1.5,sloped] \arrow[rdd, "\alpha\beta\cdot\vec v"',out=235,in=210,looseness=1.8,sloped,pos=0.4] & \\
        {\color{red}\Forg{\Span{1}}} \arrow[r, "\Forg{\Span\eta}", red, dashed,bend left] \arrow[r, "\eta"', red, dashed] \arrow[d, "\Forg{(\alpha\beta\cdot \Span{\vec v})}", red, dashed]       & {\color{red}\Forg{\Span{\Forg{\Span{1}}}}} \arrow[d, "\Forg{(\alpha\cdot \Span{\Forg{(\beta\cdot \Span{\vec v})}})}", red, dashed] \\
        {\color{red}\Forg{\Span{A}}} \arrow[r, "\eta"', red, dashed] & \Forg{\Span{\Forg{\Span{A}}}}
    \end{tikzcd}
    \end{center}

  \item $(\alpha+\beta)\cdot\vec v\equiv\alpha\cdot\vec v+\beta\cdot\vec v$.
    Notice that $\alpha\cdot\vec v+\beta\cdot\vec v$ is not typable in our
    calculus, so, this equivalence will always be taken in the form
    $(\alpha+\beta)\cdot\vec v$.

  \item $\alpha\cdot(\vec v_1+\vec v_2)\equiv\alpha\cdot\vec v_1+\alpha\cdot\vec t_2$. Same as before: if $\alpha\cdot(\vec v_1+\vec v_2)$ is typable, then
    $\alpha\cdot\vec v_1+\alpha\vec v_2$ is not, and vice-versa. Hence, only one
    of this two forms will be valid at each time.

  \item $(\lambda x.\vec t)v\lra\vec t[x:=v]$. We have
      \[
          \vcenter{
            \infer[^{\s{App}}]{\Gamma,\Delta\vdash(\lambda x.\vec t) v:A}{
             \infer[^{\s{Lam}}]{\Gamma\vdash\lambda x.\vec t:B\rightarrow A}{\Delta,x:B\vdash\vec t:A}
              &
             \Delta\vdash v:B
            }
          }
          \qquad\textrm{and}\qquad
          {\Gamma,\Delta\vdash\vec t[x:=v]:A}
      \]
      Then,
      
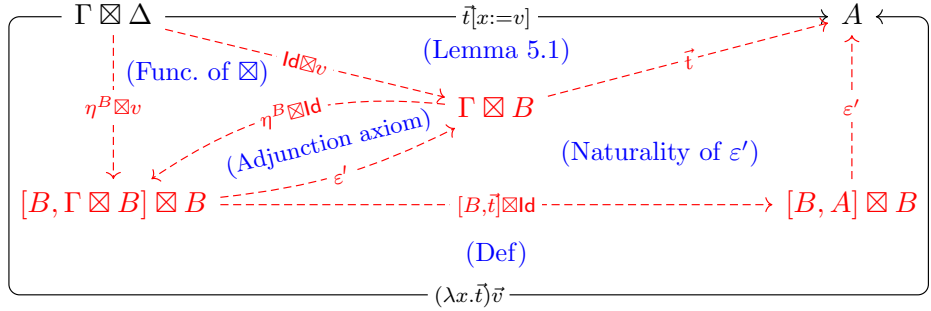
\begin{center}
        \begin{tikzcd}[
          labels=description,
          column sep=3cm,
          execute at end picture={
            \path (\tikzcdmatrixname-1-1) -- (\tikzcdmatrixname-3-1) coordinate[pos=0.5]
            (aux1) (aux1) -- (\tikzcdmatrixname-2-2) node[pos=0.23,blue,xshift=1mm,yshift=4.5mm]{\small (Func.~of~$\boxtimes$)};
            \path (\tikzcdmatrixname-1-1) -- (\tikzcdmatrixname-1-3) coordinate[pos=0.5](aux1)
            (\tikzcdmatrixname-2-2) -- (aux1) node[midway,blue]{\small (Lemma~\ref{lem:substitution})}
            (\tikzcdmatrixname-2-2) -- (aux1) node[midway,blue,yshift=-2.7cm]{\small (Def)}
            (\tikzcdmatrixname-3-1) -- (\tikzcdmatrixname-2-2) node[midway,blue,sloped,yshift=1mm]{\footnotesize (Adjunction axiom)}
            (\tikzcdmatrixname-2-2) -- (\tikzcdmatrixname-3-3) node[midway,blue]{\small (Naturality of $\varepsilon'$)};
          }]
          \Gamma\boxtimes\Delta\arrow["(\lambda x.\vec t)\vec v",
          rounded corners,
          to path={
            -| ([yshift=-3.7cm,xshift=-.7cm]\tikztostart.west)
            -- ([yshift=-3.7cm,xshift=.7cm]\tikztotarget.east)\tikztonodes
            |- (\tikztotarget)}
          ]{rr}
          \arrow[dd, "\eta^B\boxtimes v",dashed,red] \arrow[rr, "{\vec t[x:=v]}"] \arrow[rd,sloped,"\Id\boxtimes v",dashed,red] & & A \\
          &\color{red} \Gamma\boxtimes B \arrow[ld,sloped,"\eta^B\boxtimes\Id"',bend right=20,dashed,red] \arrow[ru,sloped,"\vec t",dashed,red] &  \\
          \color{red}\home B{\Gamma\boxtimes B}\boxtimes B\arrow[rr,dashed,red,"\home{B}{\vec t}\boxtimes\Id"]\arrow[ru,sloped,"\varepsilon'", bend right=10,dashed,red] & &
          \color{red}\home BA\boxtimes B \arrow[uu,dashed,red,"\varepsilon'"]
        \end{tikzcd}
      \end{center}

    \item $\Void;\vec s\lra\vec s$. We have
      \[
        \vcenter{ \infer[^{\s{PureSeq}}]{\emptyset,\Delta\vdash\Void;\vec s:A}{\infer[^{\s{Void}}]{\vdash\Void:\mathbb U}{} & \Delta\vdash\vec s:A} }
        \qquad\textrm{and}\qquad \emptyset,\Delta\vdash\vec s:A
      \]
      We write the $\emptyset$ to stress the fact that there is a hidden $\times
      1$, which can be projected out.

      Then,
      \begin{center}
        \begin{tikzcd}[labels=description,column sep=3cm,row sep=0mm,
           execute at end picture={
             \path (\tikzcdmatrixname-2-1) -- (\tikzcdmatrixname-2-3) node[midway,yshift=10mm,blue]{\small (Def)}
             (\tikzcdmatrixname-2-1) -- (\tikzcdmatrixname-2-3) node[midway,blue]{\small (Naturality of $\pi$)}
             (\tikzcdmatrixname-2-1) -- (\tikzcdmatrixname-2-3) node[midway,yshift=-10mm,blue]{\small (Def)}
             ;}
          ]
          &\color{red}1\boxtimes A\ar[rd,"\pi_A",sloped,dashed,red] &\\
          1\boxtimes\Delta\ar[rr,"\Void;\vec s",sloped,bend left=30]\ar[rr,"\vec s_{1\boxtimes\Delta}",sloped,bend right=30]\ar[ru,"\Id\boxtimes\vec s_\Delta",sloped,dashed,red]\ar[rd,"\pi_\Delta",sloped,dashed,red] && A\phantom{1\boxtimes}\\
          &\color{red}\Delta\ar[ru,"\vec s_{\Delta}",sloped] &
        \end{tikzcd}
      \end{center}

    \item $\LetP{x}{y}{\Pair{v}{w}}{\vec{s}}\lra \vec{s}[x:=v,y:=w]$. We have
      \[
        \vcenter{
          \infer[^{\s{PureLet}}]{\Gamma_1,\Gamma_2,\Delta\vdash\LetP xy{\Pair vw}{\vec s}:C}
          {
            \infer[^{\s{Pair}}]{\Gamma_1,\Gamma_2\vdash\Pair vw:A\times B}{\Gamma_1\vdash v:A & \Gamma_2\vdash w:B}
            &
            \Delta,x:A,y:B\vdash\vec s:C
          }
        }
      \]
      and
      \[
        \Gamma_1,\Gamma_2,\Delta\vdash\vec s[x:=v,y:=w]:C
      \]
      Then,
      \begin{center}
        \begin{tikzcd}[labels=description,row sep=11mm,execute at end picture={
            \path
            (\tikzcdmatrixname-1-1) -- (\tikzcdmatrixname-1-3) node[midway,yshift=5mm,blue]{\footnotesize(Lemma~\ref{lem:substitution})}
            (\tikzcdmatrixname-3-1) -- (\tikzcdmatrixname-3-3) node[midway,yshift=-5mm,blue]{\footnotesize(Def)}
            (\tikzcdmatrixname-1-1) -- (\tikzcdmatrixname-3-1) coordinate[midway](aux1)
            (\tikzcdmatrixname-1-3) -- (\tikzcdmatrixname-3-3) coordinate[midway](aux2)
            (aux1) -- (aux2) node[pos=0.8,blue]{\footnotesize(Naturality of $\varepsilon$)}
            (aux1) -- (aux2) node[pos=0.2,blue,yshift=1cm]{\footnotesize(Naturality of $\boxtimes$)}
            (\tikzcdmatrixname-1-2) -- (\tikzcdmatrixname-2-2) coordinate[midway](aux)
            (aux) -- (\tikzcdmatrixname-1-3) node[pos=0.4,blue,sloped]{\footnotesize(Lemma~\ref{lem:substitution})}
            (\tikzcdmatrixname-3-1) -- (\tikzcdmatrixname-2-2) node[midway,blue,sloped]{\footnotesize(Adjunction axiom)}
            ;
          }]
          \Gamma_1\boxtimes\Gamma_2\boxtimes\Delta
          \arrow[rr, "{\vec s[x:=v,y:=w]}",
          rounded corners,
          to path={
            -- ([yshift=5mm]\tikztostart.north)
            -- ([yshift=5mm]\tikztotarget.north)\tikztonodes
            -- (\tikztotarget)}
          ]
          \arrow[rr, "{\mathsf{let}~(x,y)=(v,w)~\mathsf{in}~\vec s}",
          rounded corners,
          to path={
            -| ([yshift=-45mm,xshift=-17mm]\tikztostart.west)
            -- ([yshift=-45mm,xshift=17mm]\tikztotarget.east)\tikztonodes
            |- (\tikztotarget)}
          ]
          \arrow[dd, "v\boxtimes w\boxtimes\eta^{A\boxtimes B}", dashed,red]
          \arrow[r, "v\boxtimes\Id", dashed,red]
          &\color{red} A\boxtimes\Gamma_2\boxtimes\Delta \arrow[r, "{\vec s[y:=w]}", dashed,red] \arrow[d, "\Id\boxtimes w\boxtimes\Id", dashed,red] & C \\
          &\color{red} A\boxtimes B\boxtimes\Delta \arrow[ld, "\Id\boxtimes\eta^{A\boxtimes B}", dashed,red,sloped,bend right] \arrow[ru,sloped,"\vec s", dashed,red] & \\
          \color{red}{(A\boxtimes B)\boxtimes[A\boxtimes B,\Delta\boxtimes
            A\boxtimes B]} \arrow[rr, "{\Id\boxtimes[A\boxtimes B,\vec s]}", dashed,red] \arrow[ru,sloped,"\varepsilon", dashed,red,out=10,in=-90] & &\color{red} {(A\boxtimes B)\boxtimes[A\boxtimes B,C]} \arrow[uu, "\varepsilon", dashed,red]
        \end{tikzcd}
      \end{center}

     \item $\Match{\Inl{v}}{x_1}{\vec{s}_1}{x_2}{\vec{s}_2}\lra
       \vec{s}_1[x_1:=v]$. We have
       \[
         \vcenter{\infer[^{\s{PureMatch}}]{\Gamma,\Delta\vdash\Match{\Inl v}{x_1}{\vec s_1}{x_2}{\vec s_2}:C}
           {
             \infer[^{\s{InL}}]{\Gamma\vdash\Inl v:A+B}{\Gamma\vdash v:A}
             &
             \Delta,x_1:A\vdash\vec s_1:C
             &
             \Delta,x_2:B\vdash\vec s_2:C
           }
         }
       \]
       and
       \[
         \Gamma,\Delta\vdash\vec s_1[x_1:=v]:C
       \]

       Then,
       \begin{center}
         \begin{tikzcd}[labels=description,column sep=2cm,row sep=11mm,execute
           at end picture={
             \path
             (\tikzcdmatrixname-1-1) -- (\tikzcdmatrixname-1-3) coordinate[midway](aux)
             (aux) -- (\tikzcdmatrixname-2-2) node[midway,blue]{\footnotesize(Lemma~\ref{lem:substitution})}
             (aux) -- (\tikzcdmatrixname-2-2) node[pos=3,blue]{\footnotesize(Def)}
             (\tikzcdmatrixname-2-2) -- (\tikzcdmatrixname-3-2) coordinate[midway](aux1)
             (aux1) -- (\tikzcdmatrixname-1-3) node[midway,blue,sloped]{\footnotesize(Def)}
             (\tikzcdmatrixname-1-1) -- (aux1) node[midway,blue,sloped]{\footnotesize(Def)}
             ;
           }]
           \Gamma\boxtimes\Delta
           \arrow[rr, "{\Match{\Inl{v}}{x_1}{\vec{s}_1}{x_2}{\vec{s}_2}}",
           rounded corners,
           to path={
             -- ([yshift=-45mm]\tikztostart.south)
             -- ([yshift=-45mm]\tikztotarget.south)\tikztonodes
             -- (\tikztotarget)}
           ]
           \arrow[rr, "{\vec{s}_1[x_1:=v]}"]
           \arrow[rdd, "\Inl{v}\boxtimes\Id",sloped,dashed,red]
           \arrow[rd, "v\boxtimes\Id",sloped,dashed,red] & & C \\
           & \color{red}A\boxtimes\Delta \arrow[d, "i_1\boxtimes\Id", dashed,red] \arrow[ru, "\vec{s}_1",sloped,dashed,red] &   \\
           & \color{red}(A+B)\boxtimes\Delta \arrow[ruu, "{[\vec{s}_1,\vec{s}_2]_1}",sloped,dashed,red] &
         \end{tikzcd}
       \end{center}

     \item $\Match{\Inr{v}}{x_1}{\vec{s}_1}{x_2}{\vec{s}_2}\lra
       \vec{s}_2[x_2:=v]$. Analogous to previous case.

     \end{itemize}

     The inductive cases are straightforward. However, we have to check the
     notation for linear constructions, which give us the typing derivations using
     the ``Unitary'' rules.

     Let $\vec v=\sum_{j=1}^n\alpha_j\cdot v_j$, $\vec w=\sum_{k=1}^m\beta_k\cdot w_k$, $\vec t=\sum_{h=1}^p\gamma_h\cdot t_h$, and
     $\vec s=\sum_{\ell=1}^q\delta_\ell\cdot s_\ell$.
      \begin{itemize}
     \item $(\vec v,\vec w):=\sum_{jk}\alpha_j\beta_k\cdot(v_j,w_k)$.

       We have
       \[
         \vcenter{
           \infer[^\sharp]{\vdash(\vec v,\vec w):\sharp(A\times B)}{
             \infer[^{\s{Pair}}]{\vdash(\vec v,\vec w):A\times B}
             {
               \vdash\vec v:A & \vdash w:B
             }
           }
         }
       \]
       and
       \[
         \vcenter{
           \infer[^{\s{Sup}}]
           {\vdash\sum_{jk}\alpha_j\beta_k\cdot(v_j,w_k):\sharp(A\times B)}
           {
             {\scriptstyle(j\neq h\vee k\neq\ell)}
             &
             \infer[^{\s{Pair}}]
             {\vdash((v_j,w_k)\perp(v_h,w_\ell)):A\times B}
             {
               \vdash v_j:A
               &
               \vdash w_k:B
             }
             & \sum_{jk}|\alpha_j\beta_k|^2=1
           }
         }
       \]

       \begin{center}
         \begin{tikzcd}[column sep=1.2cm,labels=description]
           1\boxtimes 1\arrow[rrr,"\vec v\boxtimes\vec w",dashed,red]\arrow[ddd,"\eta",dashed,red]\ar[dr,no head,dotted,blue]
           \arrow[dddrrr, "{\sum_{jk}\alpha_j\beta_k\cdot(v_j,w_k)}",
           rounded corners,
           to path={[pos=.75]
             -- ([xshift=-5mm]\tikztostart.west)
             |- ([yshift=-3mm]\tikztotarget.south)\tikztonodes
             -- (\tikztotarget)}
           ]
           \arrow[dddrrr, "{(\vec v,\vec w)}",
           rounded corners,
           to path={[pos=.25]
             -- ([yshift=3mm]\tikztostart.north)
             -| ([xshift=3mm]\tikztotarget.east)\tikztonodes
             -- (\tikztotarget)}
           ]
           & & & {\color{red}A\boxtimes B}\arrow[ddd,"\eta",dashed,red]\ar[dl,no head,dotted,blue]\\
           & {\color{blue}(*,*)}\arrow[r,maps to,blue]\arrow[d,maps to,blue] & {\color{blue}(\sum_j\alpha_j\cdot a_j,\sum_k\beta_k\cdot b_k)}\arrow[d,maps to, blue]\\
           & {\color{blue}(*,*)}\arrow[r,maps to,blue] & {\color{blue}\sum_{jk}\alpha_j\beta_k\cdot(a_j,b_k)}\\
           {\color{red}US(1\boxtimes 1)}\ar[ur,no head,dotted,blue]\arrow[rrr,"U\sum_{jk}\alpha_j\beta_k\cdot S(v_j\boxtimes w_k)",dashed,red] & &&US(A\boxtimes B)\ar[ul,no head,dotted,blue]
         \end{tikzcd}
       \end{center}

     \item $\Inl{\vec v}:=\sum_{j=1}^n\alpha_j\cdot\Inl{v_j}$.

       We have
       \[
         \vcenter{\infer[^\sharp]{\vdash\Inl{\vec v}:\sharp(A+B)}
         {
           \infer[^{\mathsf{InL}}]{\vdash\Inl{\vec v}:A+B}{\vdash\vec v:A}
         }
         }
         \qquad\textrm{and}\qquad
         \vcenter{
         \infer[^{\mathsf{Sup}}]{\vdash\sum_j\alpha_j\cdot\Inl{v_j}:\sharp(A+B)}
         {
           {\scriptstyle(h\neq k)} &
           \infer[^{\s{InL}}]{\vdash (v_h\perp v_k):A+B}{\vdash v_j:A}
           & \sum_j|\alpha_j|^2=1
         }
         }
       \]
       \begin{center}
         \begin{tikzcd}[column sep=1cm,labels=description]
           1\arrow[rr,"\vec v",dashed,red]\arrow[ddd,"\eta",dashed,red]\ar[dr,no head,dotted,blue]
           \arrow[dddrrrr, "{\sum_j\alpha_j\cdot \Inl{v_j}}",
           rounded corners,
           to path={[pos=.75]
             -- ([xshift=-5mm]\tikztostart.west)
             |- ([yshift=-3mm]\tikztotarget.south)\tikztonodes
             -- (\tikztotarget)}
           ]
           \arrow[dddrrrr, "\Inl{\vec v}",
           rounded corners,
           to path={[pos=.25]
             -- ([yshift=3mm]\tikztostart.north)
             -| ([xshift=3mm]\tikztotarget.east)\tikztonodes
             -- (\tikztotarget)}
           ]
           & & {\color{red}A}\arrow[rr,"i_1",dashed,red] & & {\color{red}A+B}\arrow[ddd,"\eta",dashed,red]\ar[dl,no head,dotted,blue]\\
           & {\color{blue}*}\arrow[r,maps to,blue]\arrow[d,maps to,blue] &
           {\color{blue}\sum_j\alpha_j\cdot a_j}\arrow[r,maps to,blue]& {\color{blue}\Inl{\sum_j\alpha_j\cdot a_j}}\arrow[d,maps to, blue]\\
           & {\color{blue}*}\arrow[rr,maps to,blue] & & {\color{blue}\sum_{j}\alpha_j\cdot\Inl{a_j}}\\
           {\color{red}US1}\ar[ur,no
           head,dotted,blue]\arrow[rrrr,"U\sum_{j}\alpha_j\cdot S\Inl{v_j}",dashed,red] &&&&US(A+B)\ar[ul,no head,dotted,blue]
         \end{tikzcd}
       \end{center}
     \item $\Inr{\vec v}:=\sum_{j=1}^n\alpha_j\cdot\Inr{v_j}$. Analogous to
       previous case.
     \item $t\vec s:=\sum_{\ell=1}^q\delta_\ell\cdot ts_\ell$.
       In this case, there is only one way to type it, which is first doing the
       $\s{Sup}$ to type $\vec s$, and then the $\s{App}$, because we cannot
       apply $\s{Sup}$ on $ts_\ell$, since these are not values.

     \item $\vec t;\vec s:=\sum_{h=1}^p\gamma_h\cdot (t_h;\vec s)$. Since
       $\Gamma\vdash\vec t:\sharp\U$, because of the orthogonality restriction
       on rule $\s{Sup}$ the only possibility is $p=1$. Thus, $\gamma_1=\gamma$
       and $t_1=t$, with $|\gamma_1|^2=1$.

       Furthermore, in this case there is only one way to type it, which is first doing the
       $\s{Sup}$ to type $\vec t$, and then the $\s{UnitarySeq}$, because we
       cannot use the rule $\s{Sup}$ on $t;\vec s$, which is not a value.

     \item $\LetP xy{\vec t}{\vec s}:=\sum_{h=1}^p\gamma_h\cdot(\LetP xy{t_h}{\vec s})$.
       In this case, there is only one way to type it, which is first doing the
       $\s{Sup}$ to type $\vec t$, and then the $\s{UnitaryLet}$, because we cannot superpose
       $\mathtt{let}$ terms since these are not values.

     \item $\Match{\vec t}{x_1}{\vec s_1}{x_2}{\vec s_2}:=\sum_{h=1}^p\gamma_h\cdot(\Match{t_h}{x_1}{\vec s_1}{x_2}{\vec s_2})$.
       In this case, there is only one way to type it, which is first doing the
       $\s{Sup}$ to type $\vec t$, and then the $\s{UnitaryMatch}$, because we cannot superpose
       $\mathtt{match}$ terms since these are not values.
        \qedhere
   \end{itemize}
 \end{proof}

\end{document}